\documentclass[11pt]{article}
\usepackage[margin=1in]{geometry}
\usepackage{amsmath,amssymb,amsfonts,amsthm,mathtools,bm}
\usepackage{graphicx}
\usepackage{booktabs}
\usepackage{float}
\usepackage{hyperref}
\usepackage{enumitem}
\usepackage{algorithm}
\usepackage{algpseudocode}
\usepackage{setspace}
\usepackage[numbers]{natbib}
\usepackage{caption}
\usepackage{subcaption}
\usepackage{xcolor}
\usepackage{array}
\usepackage{longtable}

\hypersetup{colorlinks=true,linkcolor=blue,citecolor=blue,urlcolor=blue}

\newtheorem{theorem}{Theorem}[section]
\newtheorem{proposition}[theorem]{Proposition}
\newtheorem{lemma}[theorem]{Lemma}
\newtheorem{corollary}[theorem]{Corollary}
\theoremstyle{definition}
\newtheorem{definition}[theorem]{Definition}
\newtheorem{assumption}[theorem]{Assumption}
\theoremstyle{remark}
\newtheorem{remark}[theorem]{Remark}

\begin{document}
\begin{center}
\textbf{\Large Storage Options and Endogenous Commodity Prices in Continuous Time}
 \end{center}
\begin{center}
 N. Karimi \footnote{Corresponding Author, \\ Email: nkarimi@aut.ac.ir}, E. Salavati $^{a}$, H. Adibi $^{b}$\\
\end{center}

\begin{center}
\emph{\footnotesize $^{1,a,b}$ Department of Applied Mathematics, Faculty of Mathematics and Computer Science} \\
\emph{\footnotesize Amirkabir University of Technology, No. 424, Hafez Ave.,15914, Tehran, Iran}\\
\end{center}

\begin{abstract}
We study the price formation of a storable commodity when the decision to sell or keep the commodity is treated as an embedded storage option. The price process is not imposed exogenously. Instead, a candidate price function determines the demand dynamics, while the optimal stopping value generated by those dynamics produces a new price function. We call a fixed point of this feedback mechanism an \emph{endogenous storage-option equilibrium} (ESOE), using a separate acronym to distinguish the continuous-time construction studied here from the stationary rational expectations equilibrium terminology of the classical competitive-storage literature. A global existence result is established on a compact admissible class under a transparent drift--discount condition. The argument avoids the interval-by-interval extension of local fixed points and does not claim uniqueness from Schauder's theorem. We then characterize the equilibrium as a nonlinear obstacle problem and develop a monotone finite-difference/policy-iteration method with an asymptotically consistent far-field boundary condition. The numerical analysis separates fixed-point convergence from spatial discretization error and is supplemented by Monte Carlo validation, scaling diagnostics, comparative statics, and the distribution of optimal selling times. The results clarify how demand elasticity, discounting, depreciation, and volatility jointly determine the storage premium and the selling threshold.
\end{abstract}



\section{Introduction}

The physical nature of commodities creates a pricing mechanism that is different from the one usually adopted for financial securities. A producer or a trader who owns a storable commodity does not have to sell it immediately. The commodity can be carried to a later date, and this simple possibility introduces an intertemporal decision into the spot market. When current prices are sufficiently attractive, selling is natural. When they are not, storage creates the possibility of waiting for a more favorable market condition. In this sense, storability provides an option whose exercise decision is tied directly to the commodity price.

The classical competitive-storage literature studies this mechanism through a rational expectations equilibrium in which current availability, inventory, demand, and expected future prices are linked. The modern formulation of the storage model goes back to \citet{gustafson1958}, \citet{muth1961}, and \citet{samuelson1971}, and was developed in a particularly influential form by \citet{deaton1992,deaton1995,deaton1996,deaton2003}. In the Deaton--Laroque framework the equilibrium price is not an arbitrary stochastic process. It is determined by the economic primitives of the market and by the storage decision itself. This feature remains one of the main attractions of structural storage models, especially when they are compared with reduced-form commodity models that postulate an exogenous mean-reverting spot price.

A continuous-time formulation is attractive for two reasons. First, it permits the storage decision to be written as an optimal stopping problem and therefore makes the relation between storability and option value explicit. Second, it provides a natural route to free-boundary methods, stochastic calculus, and the numerical techniques developed for perpetual American-style claims. The difficulty is that the state dynamics and the value function are not independent. A candidate speculative price affects the demand process, while the optimal stopping value associated with that process gives a new speculative price. The equilibrium is consequently a fixed point of an optimal-stopping operator.

Methodologically, the present formulation also belongs to a broader line of recent work in which continuous-time stochastic models are coupled with explicit decision rules and computational schemes. Examples include long-memory and stochastic-time-change models for temperature-index insurance \citep{karimiShokrollahi2026}, stochastic dynamic characterization of bank insolvency regions under interacting liquidity and credit risks \citep{karimiAhmadian2026}, computational analysis of functionally generated portfolios under stochastic transaction costs \citep{karimiSalavati2026}, and stochastic Runge--Kutta approximation for processes driven by mixed fractional Brownian motion \citep{karimiShahmoradi2025}. These studies address different applications, but they share the principle that stochastic dynamics, economically meaningful constraints, and numerical implementation should be analyzed within a single coherent framework. In the storage problem considered here, that principle takes the specific form of a feedback between the endogenous state dynamics and an optimal stopping value.

A short version of this continuous-time idea, together with a contraction result and a first numerical algorithm, appeared in \citet{karimi2024}. The present article is not intended to repeat that letter. It develops the model in a substantially more complete form and addresses several points that are important when the equilibrium is used as a stand-alone mathematical and computational model. In particular, the global existence argument is rebuilt on a compact admissible class in the topology of locally uniform convergence. This removes the need to construct fixed points separately on $[x_0,N]$ and then assume that those local fixed points are mutually consistent when $N$ changes. The equilibrium operator is shown to map the admissible class into itself, and the continuity needed for a Schauder--Tychonoff argument is proved by combining stability of one-dimensional diffusions with a finite-horizon truncation of the stopping problem. The resulting theorem proves existence only. Uniqueness is kept logically separate and is attached to a contraction condition, as it should be.

There is also a terminological point that is useful in the present paper. The abbreviation SREE is already associated with the stationary rational expectations equilibrium of the Deaton--Laroque competitive-storage model. Our fixed point is closely related to that concept, but its mathematical construction is specifically a continuous-time feedback between a storage-option value and an endogenous demand diffusion. To avoid using the same abbreviation for two objects introduced in different formulations, we call the fixed point studied here an \emph{endogenous storage-option equilibrium}, abbreviated \emph{ESOE}. This is a paper-specific label rather than a claim that a new general equilibrium concept is being introduced.

The numerical part is also revised in an important way. A one-step time discretization of the Bellman relation can be useful, but it mixes three separate errors: the fixed-point iteration error, the time-discretization error, and the state-space truncation error. Moreover, an artificial zero value imposed at a moderate upper state boundary can substantially depress the storage value because a producer starting from a large inventory may still wait until the state returns to the selling region. We therefore solve, at every outer fixed-point iteration, the stationary obstacle problem associated with the current candidate price function. A monotone finite-difference discretization is combined with policy iteration, and the far-field condition is obtained from the negative characteristic exponent of the asymptotic diffusion. This gives a direct numerical approximation of the infinite-horizon stopping operator and leads to a clean error decomposition.

The expanded analysis produces several observations that are difficult to see from a small set of equilibrium curves alone. Normalization by $b$ and $b/a$ reveals an almost exact scaling collapse with respect to the linear inverse-demand parameters. The storage premium increases with volatility, while the selling boundary moves only modestly for the baseline parameter range. Discounting has a much stronger effect on the boundary: a larger discount rate makes waiting less valuable and enlarges the immediate-selling region. Monte Carlo first-passage simulations agree closely with the finite-difference values, and the distribution of the optimal selling time shows that volatility mainly spreads the timing distribution rather than simply shifting its mean.

The paper is organized as follows. Section~\ref{sec:literature} places the model in the storage and stochastic-control literature. Section~\ref{sec:model} formulates the continuous-time feedback mechanism and introduces ESOE. Section~\ref{sec:existence} proves well-posedness of the demand diffusion and gives the global existence theorem. Section~\ref{sec:freeboundary} discusses the obstacle formulation, the selling boundary, and uniqueness under a separate contraction condition. Section~\ref{sec:numerics} develops the numerical scheme and states the error decomposition. Section~\ref{sec:results} presents the numerical and statistical results. Section~\ref{sec:conclusion} concludes. Technical details of the compactness and continuity arguments and of the finite-difference operator are collected in the appendices.

\section{Related literature}\label{sec:literature}

The theory of commodity storage has a long history. The early contribution of \citet{gustafson1958} already emphasized that inventory decisions have to be considered intertemporally. Rational expectations entered the discussion through \citet{muth1961} and \citet{samuelson1971}. The competitive-storage model was subsequently developed in \citet{williamswright1991} and in the sequence of papers by \citet{deaton1992,deaton1995,deaton1996,deaton2003}. In those models the non-negativity of inventory creates a kink in the equilibrium pricing rule and produces state-dependent price dynamics. The model therefore links the observed price to scarcity, inventories, and expectations rather than treating prices as an exogenous diffusion.

The empirical literature has shown both the usefulness and the difficulty of that structural approach. \citet{michaelides2000} compare simulation-based estimators, while \citet{cafiero2011,cafiero2014} study the empirical relevance and likelihood-based estimation of the storage model. More recently, \citet{osmundsen2021} combine the competitive-storage model with a stochastic trend in a state-space specification. A particularly important recent contribution is \citet{gouel2025}, who use indirect inference and information on both prices and quantities to provide a broad empirical assessment of a rational expectations commodity-storage model. Their results reinforce the point that storage models remain empirically relevant when the structural shocks and identifying information are rich enough.

A second branch of the literature values storage facilities or storage flexibility under exogenous prices. In gas markets, \citet{chen2008} formulate storage as a stochastic control problem and develop a semi-Lagrangian numerical method, while \citet{chen2010} introduce regime switching. \citet{carmona2010} study energy-storage valuation through optimal switching, and \citet{warin2012} analyzes gas-storage hedging. \citet{cummins2018} consider multifactor L\'evy dynamics. These models are natural when the storage operator is small relative to the market and takes prices as given. The present paper has a different objective: the storage decision feeds back into the demand state and the commodity price is generated endogenously.

Continuous-time competitive-storage models are less common. \citet{trofimov2022} derives a continuous-time competitive-storage equilibrium in which no-arbitrage and no-trade conditions characterize the trading zones. Our formulation is different because it is built directly around a perpetual optimal-stopping operator whose coefficients depend on the candidate equilibrium price. The related article \citet{karimi2024} introduced this stopping/fixed-point formulation in a short format. The current study expands that framework by giving a global compactness-based existence result, an asymptotic boundary condition, a direct stationary obstacle solver, and a substantially broader numerical investigation.

The recent literature also shows renewed interest in the strategic and operational role of storage. \citet{balakin2025} study dynamic trading strategies for a storage unit in an oligopolistic market with demand shocks and a capacity constraint. Their analysis is strategic and the price effect of the storage unit is explicit. Our model remains a competitive reduced-form equilibrium of demand and storage, but the two approaches share the view that storage is not merely a passive physical asset. It changes the timing of market participation and therefore changes the price formation mechanism.

Finally, the mathematical tools used below belong to the literature on optimal stopping and free-boundary problems. We use the standard Snell-envelope interpretation of perpetual stopping values as developed, for example, in \citet{peskir2006}. The numerical obstacle problem is treated by a monotone discretization so that the discrete operator has the sign structure of an $M$-matrix. This is especially useful near a free boundary, where high nominal order is less important than monotonicity and stability. The convergence discussion follows the familiar principle that a numerical approximation of a contraction has two distinct sources of error: the error of approximating the operator and the error of terminating the fixed-point iteration.

\section{Continuous-time storage model and the ESOE concept}\label{sec:model}

This section specifies the economic feedback that drives the model and fixes the notation used in the theoretical analysis. We first construct the demand state and the storage-option operator for a given candidate price, and then define the admissible class on which the endogenous equilibrium fixed point is sought.

\subsection{Economic construction}

We work on a filtered probability space $(\Omega,\mathcal F,\{\mathcal F_t\}_{t\ge0},\mathbb{P})$ supporting a standard Brownian motion $B$. The state $X_t$ represents the amount available to the market, or equivalently the demand-state variable used to determine the fundamental price. The inverse demand function is taken to be
\begin{equation}\label{eq:P}
P(x)=(b-ax)^+, \qquad a>0,\quad b>0.
\end{equation}
The upper price level is $b$ and $b/a$ is the quantity at which the fundamental price reaches zero. This specification is deliberately kept close to the competitive-storage literature and to the original version of the model.

The first part of the storage decision can be understood without introducing the state dynamics. If $p_t$ is the market price and the holder of the commodity may either sell now or postpone the sale, the discrete-time rationality condition has the form
\[
\widetilde p_t=\max\{e^{-rt}P(X_t),\mathbb{E}_t[\widetilde p_{t+1}]\},
\]
where $\widetilde p_t=e^{-rt}p_t$. The continuous-time analogue is the Snell envelope
\begin{equation}\label{eq:snell}
\widetilde p_t=\operatorname*{ess\,sup}_{\tau\ge t}
\mathbb{E}_t\!\left[e^{-r\tau}P(X_\tau)\right].
\end{equation}
This relation describes the storage option for a given state process. The second part of the model makes the state process itself depend on the speculative price.

To avoid an ambiguity that is present when one writes $P^{-1}$ for the truncated linear function in \eqref{eq:P}, we use the following generalized inverse throughout the paper:
\begin{equation}\label{eq:Q}
Q(p)=\frac{b-\Pi_{[0,b]}(p)}{a},
\qquad \Pi_{[0,b]}(p)=\min\{b,\max\{0,p\}\}.
\end{equation}
For $0\le p\le b$, $Q(p)=(b-p)/a$. The definition remains meaningful at $p=0$, where the ordinary inverse of $(b-ax)^+$ is not single-valued.

Suppose that $h$ is a candidate speculative price function. The associated demand process is defined by
\begin{equation}\label{eq:sdeh}
 dX_t^h=\mu_h(X_t^h)\,dt+\sigma X_t^h\,dB_t,\qquad X_0^h=x\ge0,
\end{equation}
with
\begin{equation}\label{eq:drift}
\mu_h(x)=(m-\delta)x-(1-\delta)Q(h(x)).
\end{equation}
Here $\delta\in[0,1)$ is the depreciation rate, $m$ is the mean component of the supply/demand shock in the continuous-time scaling, and $\sigma>0$ controls the volatility. The term $Q(h(x))$ is the current quantity released to the market at candidate price $h(x)$. The multiplicative diffusion coefficient reflects the scale of the market and, as shown below, is also useful for preserving non-negativity.

For a fixed candidate $h$, the storage-option value is
\begin{equation}\label{eq:Th}
(\mathcal{T} h)(x)=\sup_{\tau\in\mathfrak T}
\mathbb{E}_x\!\left[e^{-r\tau}P(X_\tau^h)\right],
\end{equation}
where $\mathfrak T$ is the set of almost surely finite or infinite stopping times and the convention $e^{-r\infty}P(X_\infty)=0$ is used. The operator $\mathcal{T}$ closes the feedback loop: $h$ determines the dynamics, and the dynamics determine the new storage value $\mathcal{T} h$.

\subsection{Admissible price functions}

The fixed-point problem should be posed on a function class for which the diffusion is well defined and the stopping operator has good compactness properties. We use
\begin{equation}\label{eq:H}
\mathcal{H}=\left\{h\in C([0,\infty)):
P\le h\le b,\ h(0)=b,\ h\ \text{is non-increasing},\
|h(x)-h(y)|\le a|x-y|\right\}.
\end{equation}
The Lipschitz constant $a$ is not an arbitrary regularization. It is exactly the slope of the fundamental inverse-demand curve. Section~\ref{sec:existence} shows that, under the drift--discount condition used for existence, the stopping operator naturally preserves this bound.

Because $h\ge P$, the generalized inverse satisfies
\begin{equation}\label{eq:qbound}
0\le Q(h(x))\le x,\qquad x\ge0.
\end{equation}
Indeed, for $x\le b/a$, $h(x)\ge b-ax$ gives $Q(h(x))\le x$, and for $x>b/a$ the bound follows from $Q(h(x))\le b/a<x$.

\begin{definition}[Endogenous storage-option equilibrium]\label{def:esoe}
A function $f\in\mathcal{H}$ is called an \emph{endogenous storage-option equilibrium} (ESOE) if
\begin{equation}\label{eq:fixedpoint}
\mathcal{T} f=f.
\end{equation}
The abbreviation ESOE is used in this paper to separate the continuous-time storage-option fixed point from the SREE abbreviation traditionally used for the stationary rational expectations equilibrium of the discrete competitive-storage model.
\end{definition}

At an ESOE the candidate price used in the demand dynamics and the value generated by the optimal selling rule coincide. Consequently neither the speculative price function nor the associated demand diffusion is imposed exogenously.

\section{Well-posedness and global equilibrium existence}\label{sec:existence}

The existence result is the main theoretical change relative to the old long manuscript. The proof is carried out directly on $[0,\infty)$ in the compact-open topology. There is no need to define a separate fixed-point problem on every interval $[x_0,N]$, and therefore no consistency assumption between fixed points corresponding to different values of $N$ is required.

\subsection{The demand diffusion}

We begin by establishing that the state equation is mathematically well posed and remains in the economically meaningful non-negative state space. These properties are needed before the stopping operator can be defined and compared across different initial states.

\begin{proposition}[Well-posedness and positivity]\label{prop:sde}
Let $h\in\mathcal{H}$. Then \eqref{eq:sdeh} has a unique strong solution for every $x\ge0$. Moreover $X_t^h\ge0$ almost surely for every $t\ge0$. If $x>0$, then $X_t^h>0$ almost surely for every finite $t$.
\end{proposition}

\begin{proof}
We separate well-posedness from positivity. Because every $h\in\mathcal{H}$ takes values in $[0,b]$, the generalized inverse is simply $Q(p)=(b-p)/a$ on the range of $h$. Hence, for $x,y\ge0$,
\[
|Q(h(x))-Q(h(y))|\le \frac{1}{a}|h(x)-h(y)|\le |x-y|.
\]
It follows from \eqref{eq:drift} that
\[
|\mu_h(x)-\mu_h(y)|
\le \bigl(|m-\delta|+1-\delta\bigr)|x-y|,
\qquad x,y\ge0.
\]
Thus the drift is Lipschitz on the state space, with a constant independent of the particular $h\in\mathcal{H}$, and it has at most linear growth. The diffusion coefficient $x\mapsto\sigma x$ is also globally Lipschitz and of linear growth. If desired, one may extend $h$ to $(-\infty,0)$ by setting $h(x)=b$ there; this preserves the Lipschitz property and gives a globally Lipschitz extension of the coefficients to $\mathbb{R}$. The standard existence-and-pathwise-uniqueness theorem for SDEs therefore yields a unique non-explosive strong solution.

It remains to verify that the solution cannot leave $[0,\infty)$. From \eqref{eq:qbound}, for $x\ge0$,
\[
\mu_h(x)=(m-\delta)x-(1-\delta)Q(h(x))
\ge (m-\delta)x-(1-\delta)x=(m-1)x.
\]
Consider the comparison process
\[
dY_t=(m-1)Y_t\,dt+\sigma Y_t\,dB_t,\qquad Y_0=x.
\]
It has the explicit solution
\[
Y_t=x\exp\!\left(\left(m-1-\frac{\sigma^2}{2}\right)t+\sigma B_t\right).
\]
The drift inequality above and the common diffusion coefficient $\sigma x$ permit the one-dimensional comparison theorem to be applied up to every finite time. Consequently $X_t^{h,x}\ge Y_t$ almost surely. If $x>0$, the displayed formula gives $Y_t>0$ for every finite $t$, and therefore $X_t^{h,x}>0$ almost surely as well. If $x=0$, then $h(0)=b$ and $Q(h(0))=0$, so $\mu_h(0)=0$ and the diffusion coefficient also vanishes at zero. Hence the identically zero process is a solution, and pathwise uniqueness implies $X_t^{h,0}=0$ for all $t\ge0$. This proves both non-negativity and the stated strict positivity for positive initial data.
\end{proof}

The next lemma records an order property needed below.

\begin{lemma}[Order in the initial state]\label{lem:order}
For a fixed $h\in\mathcal{H}$ and initial states $0\le x\le y$, let $X^{h,x}$ and $X^{h,y}$ be driven by the same Brownian motion. Then
\[
X_t^{h,x}\le X_t^{h,y}\quad\text{for all }t\ge0\quad\text{a.s.}
\]
\end{lemma}

\begin{proof}
Write $X_t=X_t^{h,x}$ and $Y_t=X_t^{h,y}$, and drive both equations by the same Brownian motion. Their sample paths are continuous and, by Proposition~\ref{prop:sde}, pathwise uniqueness holds. Suppose, for contradiction, that the order is violated with positive probability. On that event define the first crossing time
\[
\tau=\inf\{t\ge0:X_t>Y_t\}.
\]
Since $X_0\le Y_0$ and both paths are continuous, any strict crossing must be preceded by a meeting. More precisely, on $\{\tau<\infty\}$ continuity gives $X_\tau=Y_\tau$. Starting at time $\tau$ from this common random state and using the same Brownian increments, both processes solve the same SDE with the same initial condition. Pathwise uniqueness then implies
\[
X_{\tau+t}=Y_{\tau+t},\qquad t\ge0,
\]
almost surely on $\{\tau<\infty\}$. This contradicts the definition of a time at which $X$ crosses strictly above $Y$. Hence no crossing can occur and $X_t^{h,x}\le X_t^{h,y}$ for all $t\ge0$ almost surely.
\end{proof}

\subsection{The stopping operator preserves the admissible class}

The following condition has a direct economic interpretation. It says that the deterministic growth component of the state is not large enough to dominate depreciation plus discounting.

\begin{assumption}\label{ass:driftdisc}
Throughout the existence result,
\begin{equation}\label{eq:condition}
m\le r+\delta.
\end{equation}
\end{assumption}

\begin{proposition}[Invariance of $\mathcal{H}$]\label{prop:invariance}
Under Assumption~\ref{ass:driftdisc}, $\mathcal{T}$ maps $\mathcal{H}$ into itself.
\end{proposition}

\begin{proof}
We verify the defining properties of $\mathcal{H}$ one by one. Since $0\le P\le b$, every discounted stopping payoff lies in $[0,b]$. Immediate stopping, $\tau=0$, is admissible, so
\[
P(x)\le (\mathcal{T}h)(x)\le b,
\qquad x\ge0.
\]
At $x=0$, Proposition~\ref{prop:sde} gives $X_t^{h,0}=0$. Therefore
\[
(\mathcal{T}h)(0)=\sup_{\tau}\mathbb{E}\!\left[e^{-r\tau}P(0)\right]
=\sup_{\tau}b\,\mathbb{E}[e^{-r\tau}]=b,
\]
where the maximum is attained by $\tau=0$.

Next take $0\le x\le y$ and construct $X^{h,x}$ and $X^{h,y}$ with the same Brownian motion. Lemma~\ref{lem:order} yields $X_t^{h,x}\le X_t^{h,y}$ for every $t$. Since $P$ is non-increasing, for every stopping time $\tau$,
\[
\mathbb{E}\!\left[e^{-r\tau}P(X_\tau^{h,x})\right]
\ge
\mathbb{E}\!\left[e^{-r\tau}P(X_\tau^{h,y})\right].
\]
Taking the supremum over the same class of stopping times on both sides proves that $\mathcal{T}h$ is non-increasing.

It remains to establish the Lipschitz constant $a$. Put
\[
D_t=X_t^{h,y}-X_t^{h,x}\ge0.
\]
Subtracting the two SDEs and using $Q(h(z))=(b-h(z))/a$ on the range of $h$ gives
\[
dD_t=\left[(m-\delta)D_t+\frac{1-\delta}{a}
\bigl(h(X_t^{h,y})-h(X_t^{h,x})\bigr)\right]dt
+\sigma D_t\,dB_t.
\]
Whenever $D_t>0$, define
\[
\theta_t=
\frac{h(X_t^{h,y})-h(X_t^{h,x})}{D_t},
\]
and set $\theta_t=0$ when $D_t=0$. Because $h$ is non-increasing and $a$-Lipschitz, $-a\le\theta_t\le0$. Thus
\[
dD_t=\left[(m-\delta)+\frac{1-\delta}{a}\theta_t\right]D_t\,dt
+\sigma D_t\,dB_t,
\]
with a drift coefficient bounded above by $m-\delta$. Applying It\^o's formula to $e^{-rt}D_t$ yields
\[
d(e^{-rt}D_t)
=e^{-rt}\left[m-\delta-r+\frac{1-\delta}{a}\theta_t\right]D_t\,dt
+\sigma e^{-rt}D_t\,dB_t.
\]
Under Assumption~\ref{ass:driftdisc}, the finite-variation term is non-positive. Let
\[
\rho_n=\inf\{t\ge0:D_t\ge n\}\wedge n.
\]
The stopped stochastic integral is a true martingale, and therefore for every bounded stopping time $\tau$,
\[
\mathbb{E}\!\left[e^{-r(\tau\wedge\rho_n)}D_{\tau\wedge\rho_n}\right]
\le D_0=y-x.
\]
Letting $n\to\infty$ and using Fatou's lemma gives the same estimate for bounded $\tau$. For a general stopping time, apply the bound to $\tau\wedge N$ and let $N\to\infty$, again using non-negativity. Hence
\begin{equation}\label{eq:superD}
\mathbb{E}[e^{-r\tau}D_\tau]\le y-x
\end{equation}
for every admissible stopping time, with the convention at infinity used in the definition of $\mathcal{T}$.

Finally, $P$ is $a$-Lipschitz. For the stopping-payoff functionals $F_x(\tau)=\mathbb{E}[e^{-r\tau}P(X_\tau^{h,x})]$ and $F_y(\tau)$, the elementary inequality $\sup_\tau F_x(\tau)-\sup_\tau F_y(\tau)\le\sup_\tau(F_x(\tau)-F_y(\tau))$ gives
\[
\begin{aligned}
0\le (\mathcal{T}h)(x)-(\mathcal{T}h)(y)
&\le \sup_\tau\mathbb{E}\!\left[e^{-r\tau}
\bigl(P(X_\tau^{h,x})-P(X_\tau^{h,y})\bigr)\right]\\
&\le a\sup_\tau\mathbb{E}[e^{-r\tau}D_\tau]\\
&\le a(y-x).
\end{aligned}
\]
Thus $\mathcal{T}h$ is bounded between $P$ and $b$, equals $b$ at zero, is non-increasing, and has Lipschitz constant at most $a$. Therefore $\mathcal{T}h\in\mathcal{H}$.
\end{proof}

\subsection{Compactness and continuity}

We equip $C([0,\infty))$ with the topology of locally uniform convergence. A convenient metric is
\[
d_{\mathrm{loc}}(f,g)=\sum_{k=1}^{\infty}2^{-k}
\frac{\sup_{0\le x\le k}|f(x)-g(x)|}{1+\sup_{0\le x\le k}|f(x)-g(x)|}.
\]
With this topology, $C([0,\infty))$ is a locally convex Fr\'echet space.

\begin{lemma}[Compactness of the admissible set]\label{lem:compact}
The set $\mathcal{H}$ is nonempty, convex, closed, and compact in the topology of locally uniform convergence.
\end{lemma}

\begin{proof}
First, $P(x)=(b-ax)^+$ belongs to $\mathcal{H}$: it satisfies $0\le P\le b$, $P(0)=b$, it is non-increasing, and its Lipschitz constant is exactly $a$. Hence $\mathcal{H}$ is nonempty.

To prove convexity, let $h,g\in\mathcal{H}$ and $\lambda\in[0,1]$. Every defining condition is preserved by $\lambda h+(1-\lambda)g$: the pointwise bounds and the value at zero are immediate, a convex combination of non-increasing functions is non-increasing, and
\[
|\lambda h(x)+(1-\lambda)g(x)-\lambda h(y)-(1-\lambda)g(y)|
\le a|x-y|.
\]
Thus $\mathcal{H}$ is convex.

Now suppose $h_n\in\mathcal{H}$ and $h_n\to h$ locally uniformly. Pointwise passage to the limit gives $P\le h\le b$ and $h(0)=b$. For $x\le y$, $h_n(x)\ge h_n(y)$ for every $n$, so $h(x)\ge h(y)$. Similarly,
\[
|h(x)-h(y)|=\lim_{n\to\infty}|h_n(x)-h_n(y)|\le a|x-y|.
\]
Hence $h\in\mathcal{H}$, proving closedness in the compact-open topology.

For compactness, take an arbitrary sequence $(h_n)\subset\mathcal{H}$. On $[0,1]$ the sequence is uniformly bounded by $b$ and equicontinuous because all functions have the common Lipschitz constant $a$. Arzel\`a--Ascoli therefore gives a subsequence converging uniformly on $[0,1]$. From that subsequence extract a further subsequence converging uniformly on $[0,2]$, and continue inductively. The diagonal subsequence converges uniformly on every $[0,K]$, hence locally uniformly on $[0,\infty)$. By the closedness just proved, its limit belongs to $\mathcal{H}$. Therefore every sequence in $\mathcal{H}$ has a convergent subsequence with limit in $\mathcal{H}$. Since the compact-open topology on $C([0,\infty))$ is metrized by $d_{\mathrm{loc}}$, sequential compactness is equivalent to compactness. This proves the claim.
\end{proof}

\begin{lemma}[Continuity of the stopping operator]\label{lem:Tcont}
Let $h_n,h\in\mathcal{H}$ and $h_n\to h$ locally uniformly. Then $\mathcal{T} h_n\to\mathcal{T} h$ locally uniformly.
\end{lemma}

\begin{proof}
Fix a compact set of initial states $[0,K]$. We first compare the state processes on a finite horizon. Since $Q$ is $1/a$-Lipschitz and $h_n\to h$ locally uniformly, the corresponding drifts satisfy
\[
\mu_{h_n}\longrightarrow\mu_h
\quad\text{locally uniformly},
\]
and, by the common Lipschitz bound defining $\mathcal{H}$, all $\mu_{h_n}$ have a Lipschitz constant and a linear-growth bound independent of $n$. Couple $X^{h_n,x}$ and $X^{h,x}$ with the same Brownian motion. For $R>K$, stop both processes when either one exits $[0,R]$. The standard Burkholder--Davis--Gundy and Gronwall estimates then give, for every $T>0$,
\[
\sup_{0\le x\le K}
\mathbb{E}\!\left[
\sup_{0\le t\le T\wedge\eta_R}
|X_t^{h_n,x}-X_t^{h,x}|^2\right]\longrightarrow0,
\]
where $\eta_R$ is the common localization time. Uniform moment bounds implied by the common linear-growth estimate show that
\[
\sup_n\sup_{0\le x\le K}\mathbb{P}(\eta_R\le T)\longrightarrow0
\qquad\text{as }R\to\infty.
\]
Consequently the coupled processes converge uniformly in probability on $[0,T]$, uniformly with respect to $x\in[0,K]$. Appendix~\ref{app:continuity} records these localization estimates explicitly.

Next define the finite-horizon stopping operator
\[
(\mathcal{T}_T h)(x)=\sup_{\tau\le T}
\mathbb{E}_x[e^{-r\tau}P(X_\tau^h)].
\]
Because $P$ is bounded and Lipschitz, the preceding path convergence implies stability of the bounded reward processes $e^{-rt}P(X_t^{h_n,x})$. One way to make the passage to optimal stopping explicit is to restrict stopping times first to a deterministic grid $0=t_0<\cdots<t_M=T$. Backward induction for the discrete Snell envelope involves only conditional expectations and the maximum operation, both stable under the coupled convergence and bounded convergence. Hence the grid-restricted values for $h_n$ converge uniformly on $[0,K]$ to those for $h$. Letting the mesh of the time grid tend to zero, continuity of the sample paths and the Lipschitz property of $P$ give the continuous-time finite-horizon values. Therefore
\[
\sup_{0\le x\le K}|(\mathcal{T}_T h_n)(x)-(\mathcal{T}_T h)(x)|\longrightarrow0.
\]

It remains to control the infinite horizon uniformly. For any stopping time $\tau$, define $\tau_T=\tau$ on $\{\tau\le T\}$ and $\tau_T=T$ on $\{\tau>T\}$. Then $\tau_T\le T$, and because $P\ge0$,
\[
\mathbb{E}[e^{-r\tau}P(X_\tau^h)]
\le \mathbb{E}[e^{-r\tau_T}P(X_{\tau_T}^h)]
+b e^{-rT}.
\]
Taking suprema and using $\mathcal{T}_T h\le\mathcal{T}h$ yields
\begin{equation}\label{eq:tail}
0\le (\mathcal{T}h)(x)-(\mathcal{T}_T h)(x)\le b e^{-rT},
\end{equation}
uniformly in $h\in\mathcal{H}$ and $x\ge0$. Hence
\[
\sup_{0\le x\le K}|\mathcal{T}h_n(x)-\mathcal{T}h(x)|
\le 2be^{-rT}
+\sup_{0\le x\le K}|\mathcal{T}_T h_n(x)-\mathcal{T}_T h(x)|.
\]
Given $\varepsilon>0$, choose $T$ so that $2be^{-rT}<\varepsilon/2$, and then choose $n$ so that the finite-horizon term is below $\varepsilon/2$. This proves uniform convergence on every compact interval, i.e. $\mathcal{T}h_n\to\mathcal{T}h$ locally uniformly.
\end{proof}

We can now state the global existence theorem.

\begin{theorem}[Existence of an ESOE]\label{thm:existence}
Suppose $a,b,r,\sigma>0$, $\delta\in[0,1)$, and $m\le r+\delta$. Then the operator $\mathcal{T}$ defined by \eqref{eq:Th} has at least one fixed point $f\in\mathcal{H}$. Consequently an endogenous storage-option equilibrium exists.
\end{theorem}

\begin{proof}
The proof is an application of the Schauder--Tychonoff fixed-point theorem, so we verify its hypotheses explicitly. The ambient space $C([0,\infty))$, endowed with the topology of locally uniform convergence, is a Hausdorff locally convex topological vector space. By Lemma~\ref{lem:compact}, $\mathcal{H}$ is nonempty, convex, and compact in this topology. Proposition~\ref{prop:invariance} shows that $\mathcal{T}$ is a self-map of this set,
\[
\mathcal{T}(\mathcal{H})\subseteq\mathcal{H},
\]
and Lemma~\ref{lem:Tcont} shows that this self-map is continuous. Schauder--Tychonoff therefore guarantees at least one $f\in\mathcal{H}$ satisfying
\[
\mathcal{T}f=f.
\]
By Definition~\ref{def:esoe}, such an $f$ is an endogenous storage-option equilibrium. No contraction property is used in this argument, so the conclusion is existence only; uniqueness is not implied by the fixed-point theorem.
\end{proof}

\begin{remark}[What the existence theorem does not prove]\label{rem:unique}
Theorem~\ref{thm:existence} proves existence, not uniqueness. This distinction is important. A Schauder-type argument cannot by itself rule out several fixed points. Uniqueness requires an additional property such as contraction, monotonicity combined with a suitable order argument, or another comparison principle. In the present paper we keep the two issues separate.
\end{remark}

\section{Free-boundary representation, selling threshold, and uniqueness}\label{sec:freeboundary}

This section translates the stopping formulation into a variational inequality and clarifies the geometry of the optimal selling rule. The analysis first treats a fixed candidate price, then identifies the single-crossing structure behind the numerical selling boundary, and finally separates the additional contraction assumption required for uniqueness from the existence argument of the preceding section.

\subsection{Obstacle problem for a fixed candidate price}

Fix $h\in\mathcal{H}$ and let $V_h=\mathcal{T} h$. The infinitesimal generator of $X^h$ is
\[
\mathcal A_h\phi(x)=\mu_h(x)\phi'(x)+\frac12\sigma^2x^2\phi''(x).
\]
The perpetual optimal-stopping value is characterized by the variational inequality
\begin{equation}\label{eq:VIh}
\min\left\{V_h-P,\; rV_h-\mathcal A_hV_h\right\}=0,
\qquad x>0,
\end{equation}
with $V_h(0)=b$ and the natural boundedness condition at infinity. In the continuation set $\mathcal{C}_h=\{V_h>P\}$,
\begin{equation}\label{eq:contODE}
\frac12\sigma^2x^2V_h''(x)+\mu_h(x)V_h'(x)-rV_h(x)=0.
\end{equation}
At an ESOE, $h=f=V_f$, so \eqref{eq:VIh} becomes the nonlinear obstacle problem
\begin{equation}\label{eq:nonlinearVI}
\min\left\{f-P,\; rf-\left[(m-\delta)x-(1-\delta)Q(f(x))\right]f'
-\frac12\sigma^2x^2f''\right\}=0.
\end{equation}
The nonlinearity is not in the diffusion term. It appears because the equilibrium price changes the drift of the state process through $Q(f)$.

For the linear inverse demand, the immediate-sale region is naturally located at low values of $x$, where the fundamental price is high. At larger quantities the current fundamental price is low and storage becomes valuable. The next proposition records a single-crossing property that supports the one-sided selling geometry used in the numerical section.

\begin{proposition}[Single-crossing property of the stopping payoff]\label{prop:threshold}
Let $h\in\mathcal{H}$ and assume $m\le r+\delta$. On the linear part of the fundamental demand curve, $0<x<b/a$, the function
\[
x\longmapsto (\mathcal A_h-r)P(x)
\]
is non-decreasing. Consequently its sign can change at most once, from non-positive to non-negative values (with the possibility of a zero interval), and the region in which an infinitesimal delay is locally preferable to immediate sale has a one-sided, single-crossing structure.
\end{proposition}

\begin{proof}
On $0<x<b/a$ the fundamental price is $P(x)=b-ax$, so $P'(x)=-a$ and $P''(x)=0$. Therefore
\[
(\mathcal A_h-r)P(x)=-a\mu_h(x)-r(b-ax).
\]
Substituting \eqref{eq:drift} and using $Q(h(x))=(b-h(x))/a$ gives
\begin{align}
(\mathcal A_h-r)P(x)
&=-a\left[(m-\delta)x-(1-\delta)\frac{b-h(x)}{a}\right]-r(b-ax)\\
&=(1-\delta)(b-h(x))-rb+a(r+\delta-m)x.\label{eq:AP}
\end{align}
Now let $0<x_1<x_2<b/a$. Subtracting the value of \eqref{eq:AP} at $x_1$ from that at $x_2$ yields
\[
\begin{aligned}
&[(\mathcal A_h-r)P](x_2)-[(\mathcal A_h-r)P](x_1)\\
&\qquad=(1-\delta)\bigl(h(x_1)-h(x_2)\bigr)
+a(r+\delta-m)(x_2-x_1).
\end{aligned}
\]
Both terms on the right are non-negative: the first because $h$ is non-increasing, and the second because $m\le r+\delta$. Thus $x\mapsto(\mathcal A_h-r)P(x)$ is non-decreasing on the linear part of $P$.

A non-decreasing function may vanish on an interval, but its sign can change at most once, and any sign change must be from negative to positive as $x$ increases. In the stopping interpretation, $(\mathcal A_h-r)P>0$ means that the discounted expected payoff has positive first-order drift under an infinitesimal continuation, whereas $(\mathcal A_h-r)P<0$ favors immediate exercise at the local level. Hence the local preference between selling and delaying can switch at most once as $x$ increases. This is the asserted single-crossing property; it does not by itself impose global connectedness of the stopping set, which is why the subsequent remark states that issue separately.
\end{proof}

\begin{remark}[One-sided selling boundary]\label{rem:threshold}
Proposition~\ref{prop:threshold} is deliberately stated as a single-crossing result rather than as an unconditional theorem that the global stopping set must be connected. A one-sided stopping boundary follows under the standard regularity and one-sidedness conditions for the corresponding one-dimensional perpetual stopping problem. This is the configuration obtained in all numerical experiments below. In that case we write
\[
\mathcal{D}_h=[0,x_h^*],\qquad \mathcal{C}_h=(x_h^*,\infty),
\]
and continuous fit gives $V_h(x_h^*)=P(x_h^*)$. Whenever the boundary is regular for the diffusion and smooth fit applies,
\begin{equation}\label{eq:smoothfit}
V_h'(x_h^*)=-a.
\end{equation}
This formulation avoids building a global connectedness claim into the existence theorem; the existence of an ESOE itself does not require a prior free-boundary geometry.
\end{remark}

For an equilibrium with the one-sided stopping geometry of Remark~\ref{rem:threshold}, we write $x^*=x_f^*$ and
\begin{equation}\label{eq:pstar}
p^*=f(x^*)=P(x^*)=b-ax^*.
\end{equation}
The optimal selling time is then
\begin{equation}\label{eq:taustar}
\tau^*=\inf\{t\ge0:X_t^f\le x^*\}.
\end{equation}
Because price is decreasing in the quantity state, this is equivalently the first time the equilibrium price reaches the selling level $p^*$ from below. The threshold $p^*$ is a threshold for the \emph{undiscounted} equilibrium price $p_t=f(X_t^f)$. If one writes the discounted process $\widetilde p_t=e^{-rt}p_t$, the corresponding discounted boundary is $e^{-rt}p^*$ and is therefore time dependent.

\subsection{Uniqueness under contraction}

The global existence theorem does not require the operator to be a contraction. When a contraction estimate is available, however, uniqueness and convergence of the outer iteration follow immediately. A sufficient contraction condition for the continuous-time storage operator was derived in the shorter study \citet{karimi2024}. Rather than repeat that proof here, we formulate the numerical consequences in a way that makes the logical dependence explicit.

\begin{assumption}[Contraction regime]\label{ass:contraction}
There is a complete metric $d_w$ on a nonempty $\mathcal{T}$-invariant subset $\mathcal{H}_w\subset\mathcal{H}$ and a constant $q\in(0,1)$ such that
\begin{equation}\label{eq:contract}
d_w(\mathcal{T} h,\mathcal{T} g)\le q\,d_w(h,g),\qquad h,g\in\mathcal{H}_w.
\end{equation}
\end{assumption}

\begin{corollary}[Uniqueness and iterative convergence]\label{cor:unique}
Under Assumption~\ref{ass:contraction}, the ESOE is unique in $\mathcal{H}_w$. Moreover, for $h^{(n+1)}=\mathcal{T} h^{(n)}$,
\begin{equation}\label{eq:qlinear}
d_w(h^{(n)},f)\le q^n d_w(h^{(0)},f).
\end{equation}
\end{corollary}

\begin{proof}
Because $(\mathcal{H}_w,d_w)$ is complete and $\mathcal{T}$ maps $\mathcal{H}_w$ into itself with contraction coefficient $q<1$, Banach's fixed-point theorem already guarantees a unique fixed point in $\mathcal{H}_w$. The uniqueness can also be seen directly: if $f$ and $g$ are two fixed points, then
\[
d_w(f,g)=d_w(\mathcal{T}f,\mathcal{T}g)\le q\,d_w(f,g).
\]
Since $1-q>0$, this implies $d_w(f,g)=0$, and therefore $f=g$.

For the iterates, the fixed-point identity $\mathcal{T}f=f$ and the contraction inequality give
\[
d_w(h^{(n+1)},f)
=d_w(\mathcal{T}h^{(n)},\mathcal{T}f)
\le q\,d_w(h^{(n)},f).
\]
Applying this estimate recursively yields
\[
d_w(h^{(n)},f)\le q^n d_w(h^{(0)},f),
\]
which is \eqref{eq:qlinear}.
\end{proof}

\begin{remark}[Meaning of ``order one'']
Equation \eqref{eq:qlinear} is geometric, or $q$-linear, convergence of the fixed-point iteration. It is different from the order of a time- or space-discretization. In particular, contraction of the exact operator does not by itself imply that a discrete approximation has an $O(\Delta t)$ or $O(\Delta x)$ consistency error. Section~\ref{sec:numerics} keeps these two notions separate.
\end{remark}

\section{Numerical approximation}\label{sec:numerics}

This section develops a numerical approximation that mirrors the analytical fixed-point construction as closely as possible. The main ingredients are a stationary obstacle solve for each outer iterate, an asymptotically consistent far-field condition, a monotone finite-difference discretization, and an explicit separation between iteration error and spatial discretization error.

\subsection{Why a stationary obstacle solver is preferable here}

A natural numerical version of the Bellman relation uses a short time step $\Delta t$ and updates
\[
V_{n+1}(x)=\max\left\{P(x),e^{-r\Delta t}\mathbb{E}[V_n(X_{t+\Delta t})\mid X_t=x]\right\}.
\]
This is useful and was used in the earlier manuscript. For the present infinite-horizon problem, however, it has two disadvantages. First, the contraction factor of a time-stepping Bellman update is close to $e^{-r\Delta t}$, so a small $\Delta t$ may require many value iterations. Second, the equilibrium iteration and the time approximation become entangled.

We instead approximate the operator $h\mapsto\mathcal{T} h$ directly. For a fixed outer iterate $h^{(k)}$, we solve the stationary obstacle problem \eqref{eq:VIh} and call the result $V^{(k+1)}$. The next outer iterate is $h^{(k+1)}=V^{(k+1)}$. In this way each outer step approximates the infinite-horizon stopping operator itself.

\subsection{State truncation and far-field condition}

Let $0=x_0<x_1<\cdots<x_N=X_{\max}$ with uniform spacing $\Delta x$. At the left boundary, $V(0)=b$. A zero Dirichlet condition at a moderate $X_{\max}$ is generally inappropriate: a large current quantity does not make the option worthless, because the state can drift down and eventually enter the selling region.

For large $x$, $h(x)$ is bounded and $Q(h(x))$ is bounded, whereas the leading drift term is $(m-\delta)x$. The continuation equation is therefore asymptotically
\[
\frac12\sigma^2x^2V''+(m-\delta)xV'-rV\approx0.
\]
Seeking $V(x)\sim Cx^\xi$ gives
\begin{equation}\label{eq:characteristic}
\frac12\sigma^2\xi(\xi-1)+(m-\delta)\xi-r=0.
\end{equation}
Let $\xi_-<0$ be the negative root,
\begin{equation}\label{eq:ximinus}
\xi_-=
\frac{-(m-\delta-\sigma^2/2)-
\sqrt{(m-\delta-\sigma^2/2)^2+2r\sigma^2}}{\sigma^2}.
\end{equation}
We impose the asymptotic Robin condition
\begin{equation}\label{eq:robin}
X_{\max}V'(X_{\max})=\xi_-V(X_{\max}).
\end{equation}
This condition preserves the slow power-law decay of the perpetual waiting value and is much less sensitive to the chosen truncation point than $V(X_{\max})=0$.

\subsection{Monotone finite differences and policy iteration}

For an outer iterate $h$, write $\mu_i=\mu_h(x_i)$ and $d_i=\frac12\sigma^2x_i^2/(\Delta x)^2$. We discretize the diffusion centrally. The first derivative is chosen so that the discrete generator has nonnegative off-diagonal coefficients:
\[
(\mathcal A_h^{\Delta x}V)_i=
\begin{cases}
 d_i(V_{i-1}-2V_i+V_{i+1})+
 \mu_i\dfrac{V_{i+1}-V_i}{\Delta x}, & \mu_i\ge0,\\[2ex]
 d_i(V_{i-1}-2V_i+V_{i+1})+
 \mu_i\dfrac{V_i-V_{i-1}}{\Delta x}, & \mu_i<0.
\end{cases}
\]
The discrete obstacle problem is
\begin{equation}\label{eq:discLCP}
\min\left\{V_i-P_i,\; rV_i-(\mathcal A_h^{\Delta x}V)_i\right\}=0,
\qquad i=1,\ldots,N-1,
\end{equation}
combined with $V_0=b$ and the first-order discretization of \eqref{eq:robin}. The matrix associated with $rI-\mathcal A_h^{\Delta x}$ has positive diagonal entries and nonpositive off-diagonal entries. The resulting linear complementarity problem is therefore well suited to policy iteration.

At each policy step, nodes are divided into a stopping set, where $V_i=P_i$, and a continuation set, where $(rI-\mathcal A_h^{\Delta x})V=0$. The corresponding tridiagonal system is solved and the active set is updated until it no longer changes. In the reported experiments the inner active-set solve becomes very short after the first few outer iterations.

The complete outer algorithm is as follows.

\begin{enumerate}[label=\textbf{Step \arabic*:},leftmargin=2.4cm]
\item Set $h^{(0)}=P$.
\item For the current $h^{(k)}$, compute $Q(h^{(k)})$ and the drift $\mu_{h^{(k)}}$ on the grid.
\item Solve the discrete obstacle problem \eqref{eq:discLCP} with policy iteration and the boundary conditions $V_0=b$ and \eqref{eq:robin}.
\item Set $h^{(k+1)}=V$ and compute the fixed-point residual $R_k=\left\lVert h^{(k+1)}-h^{(k)}\right\rVert_\infty$.
\item Stop when $R_k$ is below tolerance; otherwise return to Step 2.
\end{enumerate}

\subsection{Error decomposition}

The main advantage of the preceding construction is that the two convergence questions can be stated separately. Let $\mathcal{T}_{\Delta x}$ denote the discrete obstacle operator and let $f_{\Delta x}$ be its fixed point. Assume that the exact operator is a contraction with coefficient $q<1$ in a norm $\left\lVert \cdot\right\rVert$ and that, on the relevant admissible set,
\begin{equation}\label{eq:opconsistency}
\left\lVert \mathcal{T}_{\Delta x}h-\mathcal{T} h\right\rVert\le C\Delta x.
\end{equation}
The first-order consistency is natural for the monotone upwind discretization and also reflects the loss of smoothness at a free boundary.

\begin{proposition}[Fixed-point/discretization error split]\label{prop:error}
Suppose \eqref{eq:contract} and \eqref{eq:opconsistency} hold in the same norm, and let $f$ and $f_{\Delta x}$ be fixed points of $\mathcal{T}$ and $\mathcal{T}_{\Delta x}$, respectively. Then
\begin{equation}\label{eq:fixederror}
\left\lVert f-f_{\Delta x}\right\rVert\le \frac{C}{1-q}\Delta x.
\end{equation}
If the discrete outer iteration is itself contractive with coefficient $q_{\Delta x}<1$, then after $k$ outer iterations,
\begin{equation}\label{eq:totalerror}
\left\lVert f-h_{\Delta x}^{(k)}\right\rVert
\le \frac{C}{1-q}\Delta x
+q_{\Delta x}^{k}\left\lVert f_{\Delta x}-h_{\Delta x}^{(0)}\right\rVert.
\end{equation}
\end{proposition}

\begin{proof}
Because $f=\mathcal{T}f$ and $f_{\Delta x}=\mathcal{T}_{\Delta x}f_{\Delta x}$,
\[
\begin{aligned}
\left\lVert f-f_{\Delta x}\right\rVert
&=\left\lVert \mathcal{T}f-\mathcal{T}_{\Delta x}f_{\Delta x}\right\rVert\\
&\le \left\lVert \mathcal{T}f-\mathcal{T}f_{\Delta x}\right\rVert
+\left\lVert \mathcal{T}f_{\Delta x}-\mathcal{T}_{\Delta x}f_{\Delta x}\right\rVert.
\end{aligned}
\]
The contraction property of the exact operator bounds the first term by $q\left\lVert f-f_{\Delta x}\right\rVert$, while the operator-consistency assumption \eqref{eq:opconsistency}, applied at $h=f_{\Delta x}$, bounds the second by $C\Delta x$. Hence
\[
(1-q)\left\lVert f-f_{\Delta x}\right\rVert\le C\Delta x.
\]
Since $q<1$, division by $1-q$ gives \eqref{eq:fixederror}.

For the finite number of discrete outer iterations, insert the discrete fixed point $f_{\Delta x}$ and use the triangle inequality:
\[
\left\lVert f-h_{\Delta x}^{(k)}\right\rVert
\le \left\lVert f-f_{\Delta x}\right\rVert
+\left\lVert f_{\Delta x}-h_{\Delta x}^{(k)}\right\rVert.
\]
The first term is bounded by \eqref{eq:fixederror}. If the discrete map is contractive with coefficient $q_{\Delta x}$, then repeated application of its contraction estimate gives
\[
\left\lVert f_{\Delta x}-h_{\Delta x}^{(k)}\right\rVert
\le q_{\Delta x}^{k}
\left\lVert f_{\Delta x}-h_{\Delta x}^{(0)}\right\rVert.
\]
Combining the last two displays proves \eqref{eq:totalerror}. The two terms have different origins: the first is the spatial/operator approximation error and the second is the error caused by terminating the outer fixed-point iteration after $k$ steps.
\end{proof}

This proposition is simple, but it prevents an important interpretational mistake. The geometric decay of the outer residual and the first-order grid convergence are different facts and are measured separately below.

\section{Numerical results and statistical diagnostics}\label{sec:results}

This section evaluates the equilibrium computation from both numerical and economic perspectives. We begin with a transparent baseline specification, then examine grid and boundary sensitivity, comparative statics, demand scaling, Monte Carlo validation, and the distributional behavior of the optimal selling time.

\subsection{Baseline configuration}

The baseline parameters are chosen inside the range used in the original numerical study and are listed in Table~\ref{tab:base}. The values $a=0.05$ and $b=0.2$ give the fundamental quantity scale $b/a=4$. We normalize quantity by
\begin{equation}\label{eq:y}
y=\frac{ax}{b}=\frac{x}{b/a}.
\end{equation}
Thus $y=1$ is the point at which the fundamental price $P$ reaches zero. Unless otherwise stated, the computational domain is $0\le y\le4$, the grid contains 1601 points for the baseline solution, and the outer tolerance is $10^{-10}$.

\begin{table}[H]
\centering
\caption{Baseline economic and numerical parameters.}\label{tab:base}
\begin{tabular}{lll}
\toprule
Symbol & Value & Interpretation\\
\midrule
$a$ & 0.05 & slope of inverse demand\\
$b$ & 0.20 & maximal fundamental price\\
$r$ & 0.05 & discount rate\\
$\delta$ & 0.15 & depreciation rate\\
$m$ & 0 & deterministic shock component\\
$\sigma$ & 0.25 & state volatility\\
$X_{\max}$ & $4(b/a)=16$ & upper state truncation\\
$N+1$ & 1601 & grid points in baseline solution\\
outer tolerance & $10^{-10}$ & fixed-point stopping criterion\\
\bottomrule
\end{tabular}
\end{table}

Figure~\ref{fig:baseline} displays the fundamental price and the ESOE price. The selling boundary is $x^*\approx0.18$, or $y^*\approx0.045$, with $p^*\approx0.191$. The continuation region therefore begins while the current fundamental price is still relatively high. This is economically consistent with the strong expected decline of the quantity state under the baseline parameters: waiting can be valuable well before the fundamental price reaches zero. The difference $f-P$ is the storage-option premium. Once $y>1$, the fundamental price is zero, but the equilibrium value remains positive because the owner can wait for the state to return to the selling region.

\begin{figure}[H]
\centering
\includegraphics[width=.86\textwidth]{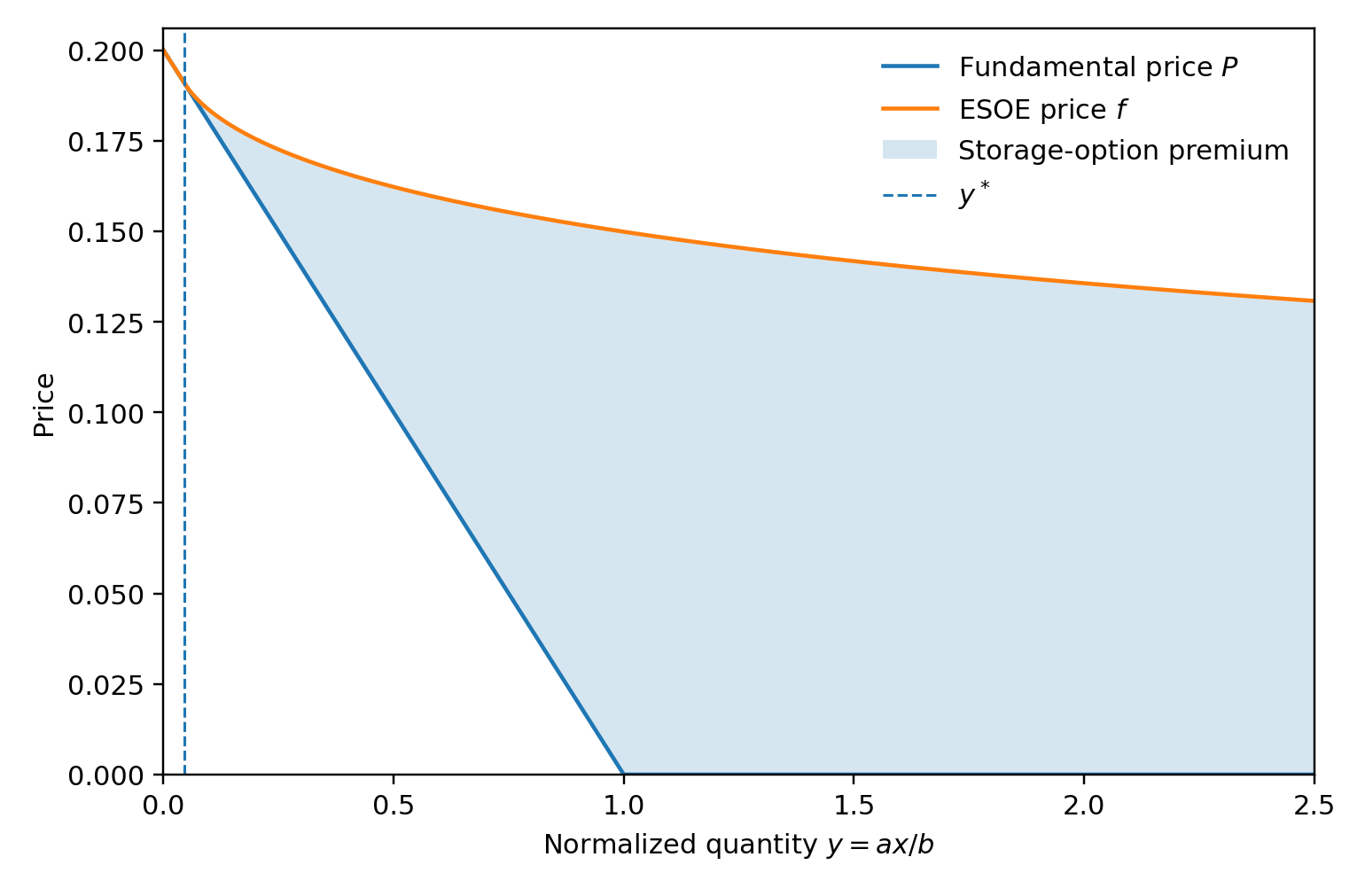}
\caption{Fundamental price, ESOE price, and storage-option premium for the baseline parameters. Quantity is normalized by $b/a$.}\label{fig:baseline}
\end{figure}

The outer fixed-point iteration is shown in Figure~\ref{fig:iterations}. Starting from $P$, the first update already generates a substantial continuation value at larger quantities. Subsequent updates change the state drift through $Q(h)$ and move the curve toward equilibrium. The iterations preserve the economically required shape: the price remains between $P$ and $b$ and is non-increasing.

\begin{figure}[H]
\centering
\includegraphics[width=.84\textwidth]{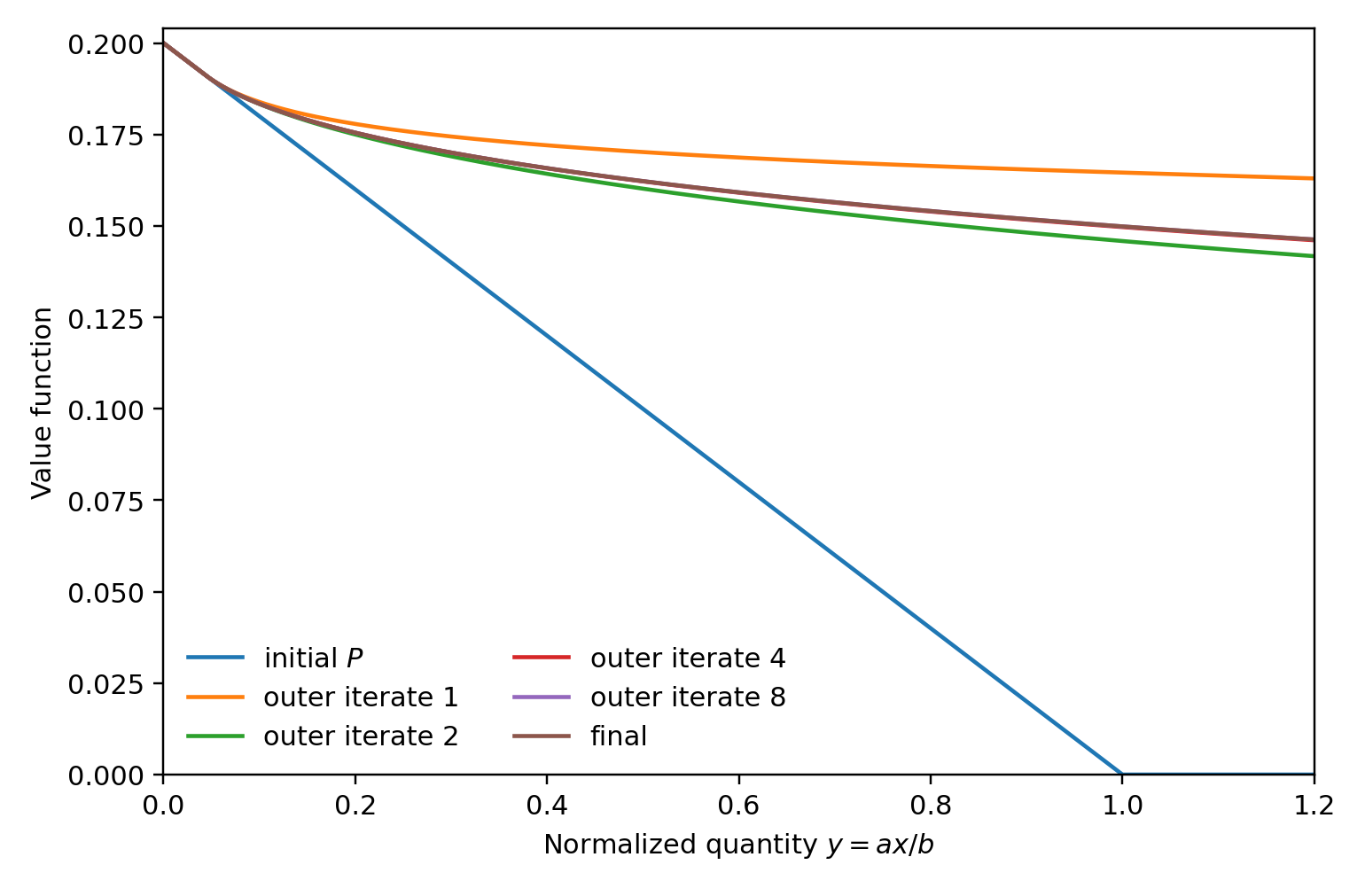}
\caption{Selected outer fixed-point iterates. The initial function is the fundamental price $P$.}\label{fig:iterations}
\end{figure}

For the baseline computation the fixed-point residual falls below $10^{-10}$ in a small number of outer updates. Figure~\ref{fig:residual} reports both the residual and the empirical ratio $R_k/R_{k-1}$. The latter should be interpreted as a numerical diagnostic of geometric convergence, not as a substitute for a theoretical contraction constant.

\begin{figure}[H]
\centering
\includegraphics[width=.84\textwidth]{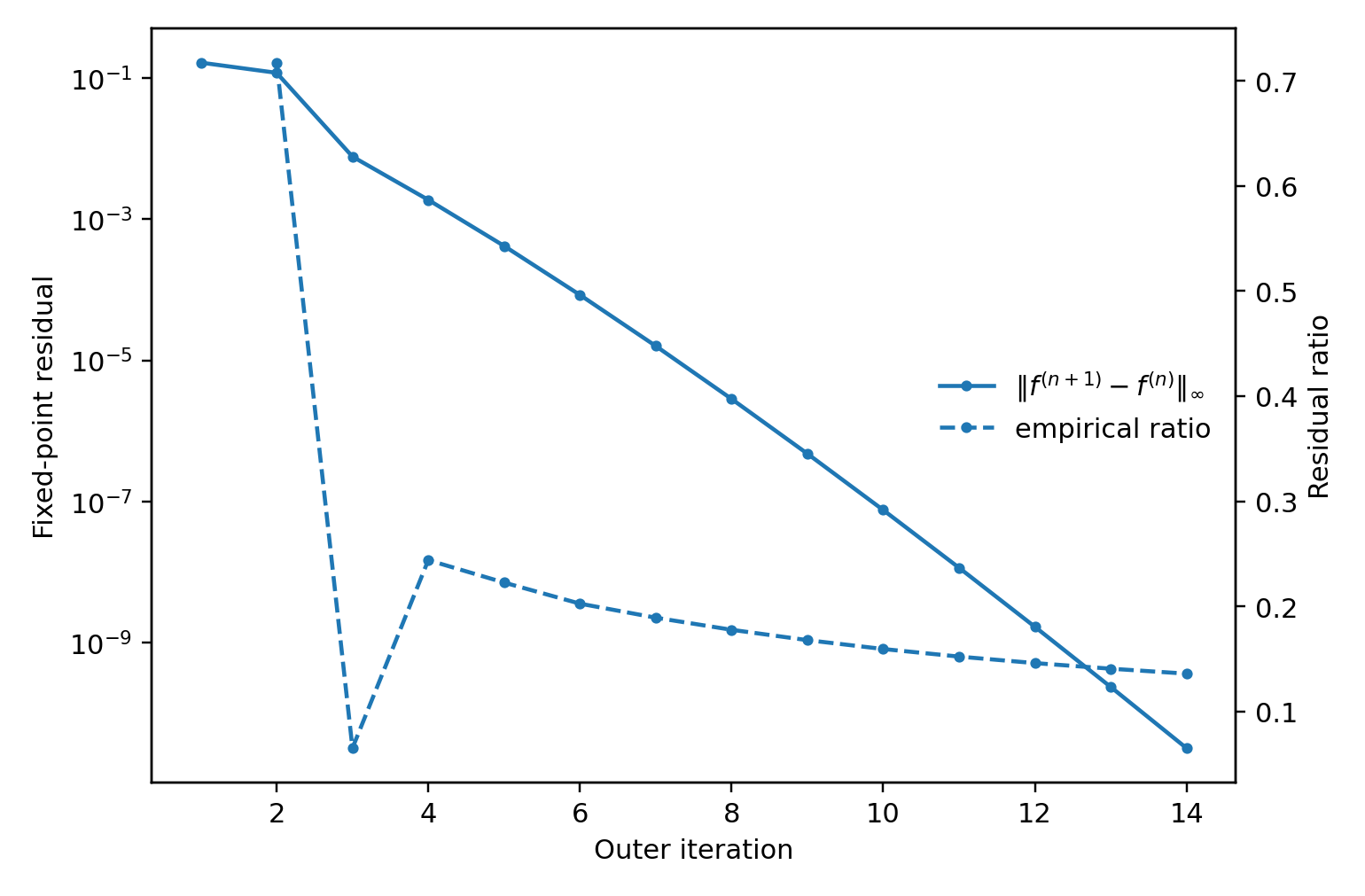}
\caption{Fixed-point residual and empirical residual ratio for the baseline problem.}\label{fig:residual}
\end{figure}

\subsection{Spatial convergence and the far-field boundary}

Table~\ref{tab:grid} compares four spatial grids. The reference is the $N+1=1601$ solution evaluated on a common set of points. The observed error decreases monotonically and the empirical rate is consistent with at least first-order behavior. The rate is somewhat larger than one between the two finer non-reference grids; this is not surprising because the solution is smooth away from the free boundary while the global discretization is only first-order monotone near convection-dominated regions and the obstacle.

\begin{table}[H]
\centering
\caption{Spatial grid convergence for the baseline problem.}\label{tab:grid}
\begin{tabular}{rrrrrr}
\toprule
Grid points & $\Delta x$ & $L^\infty$ error & EOC & $x^*$ & outer iterations\\
\midrule
201  & 0.080 & $7.20\times10^{-4}$ & --   & 0.16 & 15\\
401  & 0.040 & $3.06\times10^{-4}$ & 1.23 & 0.16 & 13\\
801  & 0.020 & $1.01\times10^{-4}$ & 1.60 & 0.18 & 13\\
1601 & 0.010 & reference & -- & 0.18 & 12\\
\bottomrule
\end{tabular}
\end{table}

\begin{figure}[H]
\centering
\includegraphics[width=.76\textwidth]{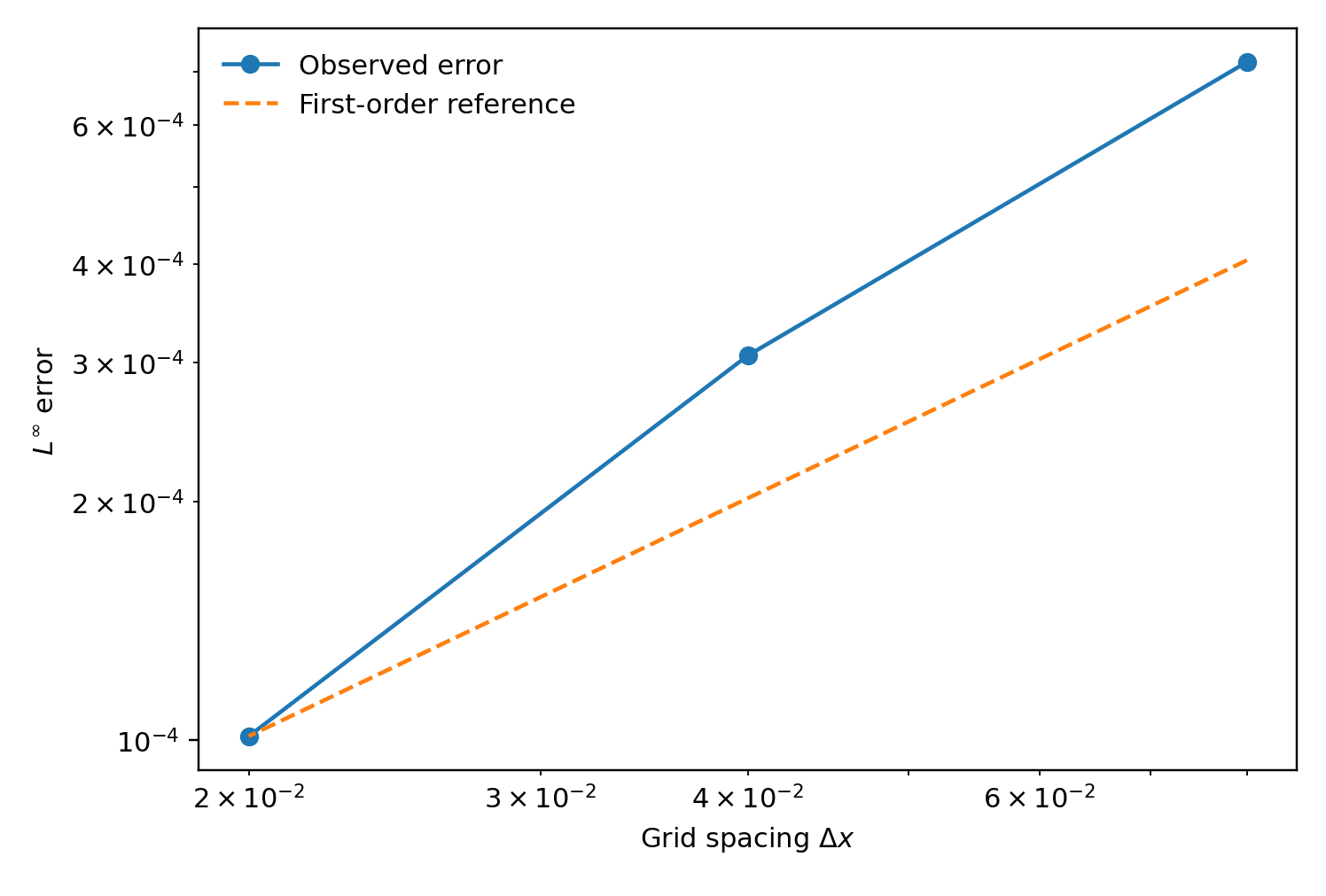}
\caption{Observed $L^\infty$ grid error compared with a first-order reference slope.}\label{fig:gridconv}
\end{figure}

The upper boundary condition is much more important than it may first appear. Figure~\ref{fig:boundary} compares the asymptotic Robin condition with a zero Dirichlet condition imposed at $y_{\max}=2$. Near the selling boundary both computations are almost indistinguishable, but the zero boundary forces the value down as the state approaches the truncation point. Table~\ref{tab:boundary} quantifies this effect against a larger-domain Robin reference. At $y=2$, the zero boundary makes the option value exactly zero even though the reference value is about $0.136$. The asymptotic boundary retains most of that value.

\begin{figure}[H]
\centering
\includegraphics[width=.84\textwidth]{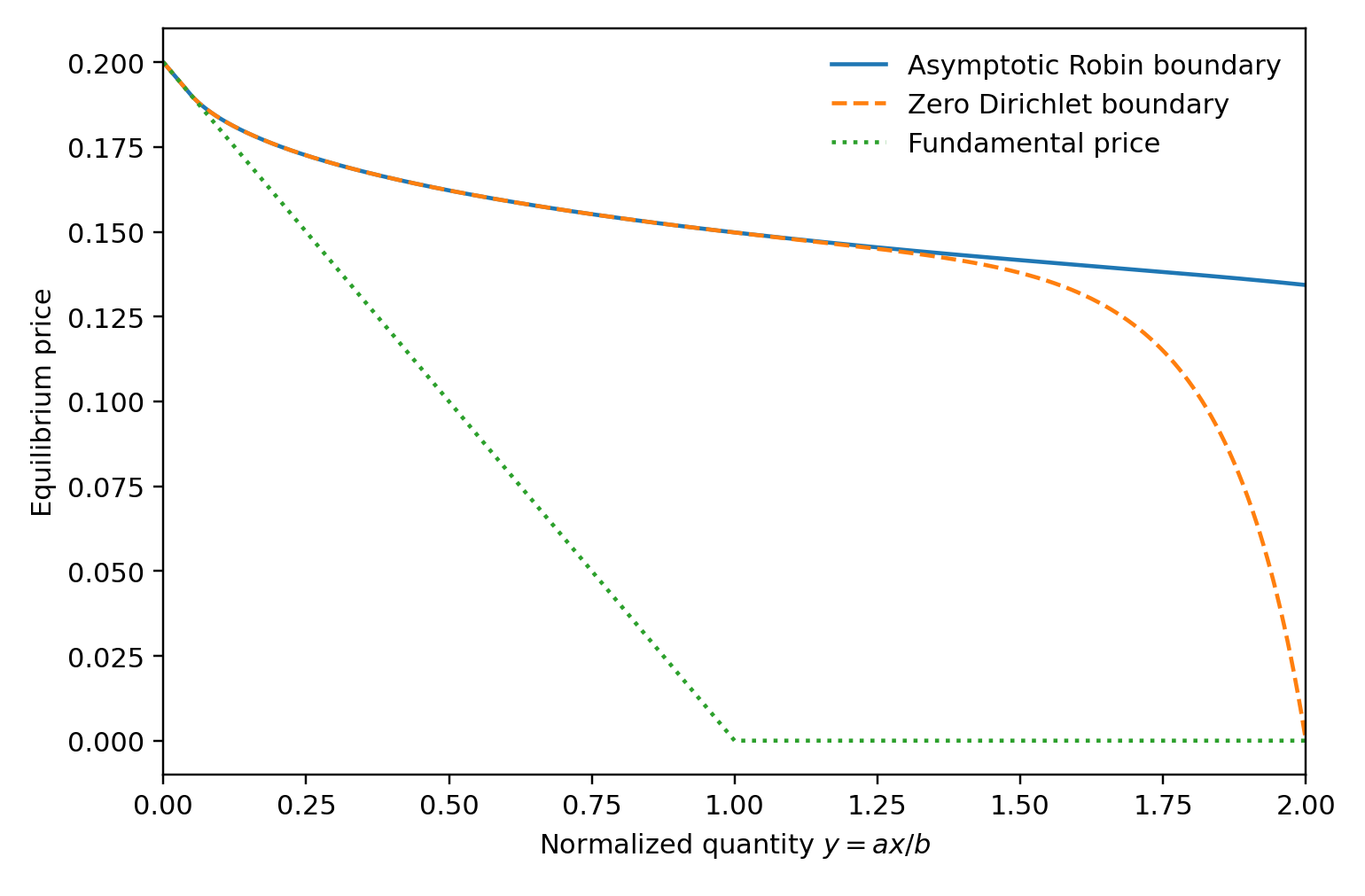}
\caption{Sensitivity to the upper state boundary. A zero Dirichlet value can materially understate the perpetual storage option.}\label{fig:boundary}
\end{figure}

\begin{table}[H]
\centering
\caption{Boundary-condition sensitivity. The reference uses the Robin condition on the larger domain $y_{\max}=4$.}\label{tab:boundary}
\begin{tabular}{rrrrrr}
\toprule
$y$ & Reference & Robin, $y_{\max}=2$ & Abs. error & Zero BC, $y_{\max}=2$ & Abs. error\\
\midrule
0.5 & 0.162240 & 0.162220 & $1.99\!\times\!10^{-5}$ & 0.162220 & $1.99\!\times\!10^{-5}$\\
1.0 & 0.149856 & 0.149837 & $1.86\!\times\!10^{-5}$ & 0.149812 & $4.35\!\times\!10^{-5}$\\
1.5 & 0.141732 & 0.141662 & $6.95\!\times\!10^{-5}$ & 0.137897 & $3.84\!\times\!10^{-3}$\\
2.0 & 0.135641 & 0.134329 & $1.31\!\times\!10^{-3}$ & 0.000000 & $1.36\!\times\!10^{-1}$\\
\bottomrule
\end{tabular}
\end{table}

This calculation also explains why a numerical equilibrium can appear to approach zero too quickly when a short state interval is used. The perpetual option has a slow power-law tail, and that tail is part of the economics of waiting. State truncation should therefore be treated as a numerical approximation rather than as an implicit economic assumption that the option becomes worthless above a selected inventory level.

\subsection{Volatility, discounting, depreciation, and drift}

Figure~\ref{fig:sigma} shows equilibrium curves for three values of $\sigma$. Higher volatility raises the storage value, particularly in the continuation region. The result is consistent with the option interpretation, but the endogenous state dynamics make the comparison less trivial than a direct appeal to a standard American put: changing $\sigma$ also changes the distribution of the demand state under the equilibrium price function.

\begin{figure}[H]
\centering
\includegraphics[width=.84\textwidth]{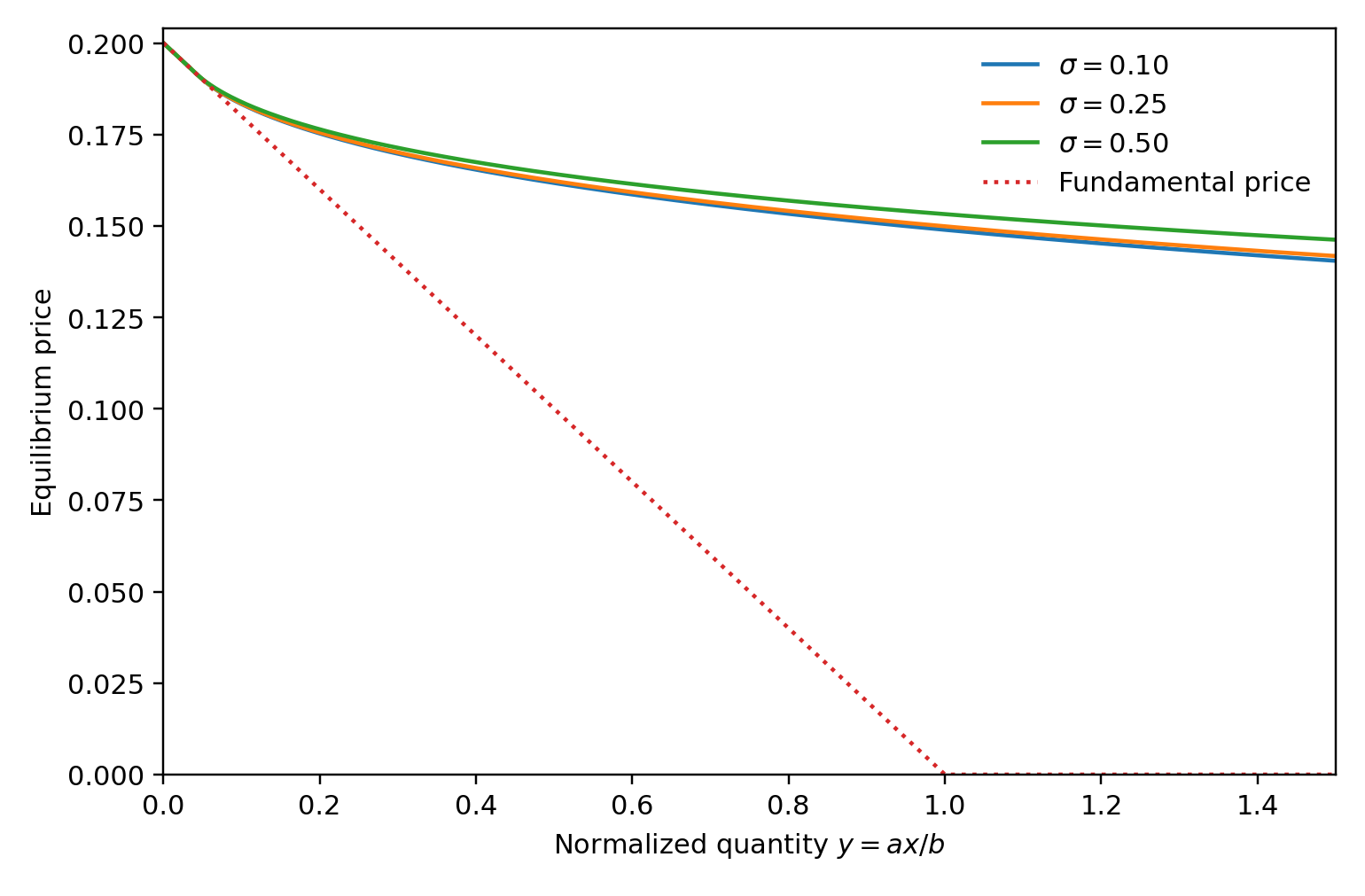}
\caption{Volatility sensitivity of the ESOE price.}\label{fig:sigma}
\end{figure}

The selling boundary and normalized selling price are shown over a wider volatility range in Figure~\ref{fig:sigmathreshold}. The boundary changes in small discrete steps because it is identified on the numerical grid. The more robust pattern is the increase of the storage premium as volatility rises. At $y=1$, the normalized premium rises from approximately $0.744$ at $\sigma=0.05$ to approximately $0.791$ at $\sigma=0.75$.

\begin{figure}[H]
\centering
\includegraphics[width=.82\textwidth]{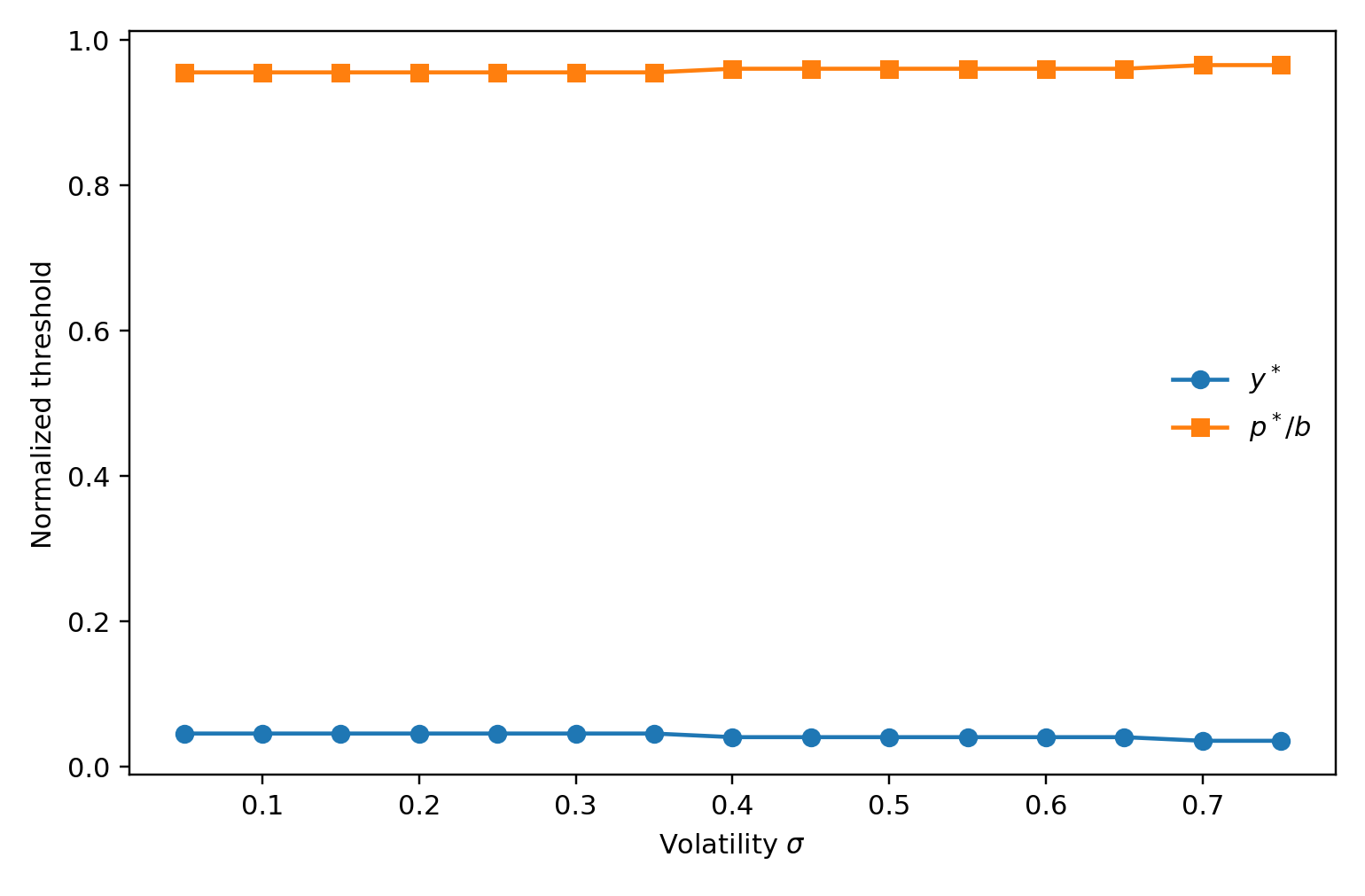}
\caption{Normalized selling boundary $y^*$ and selling price $p^*/b$ as functions of volatility.}\label{fig:sigmathreshold}
\end{figure}

Figure~\ref{fig:premiumheat} gives a two-dimensional view of the same effect. The option premium is zero in the immediate-sale region and rises sharply once the state moves into the continuation region. For $y\ge1$, all value comes from the option to wait because the contemporaneous fundamental price is zero.

\begin{figure}[H]
\centering
\includegraphics[width=.84\textwidth]{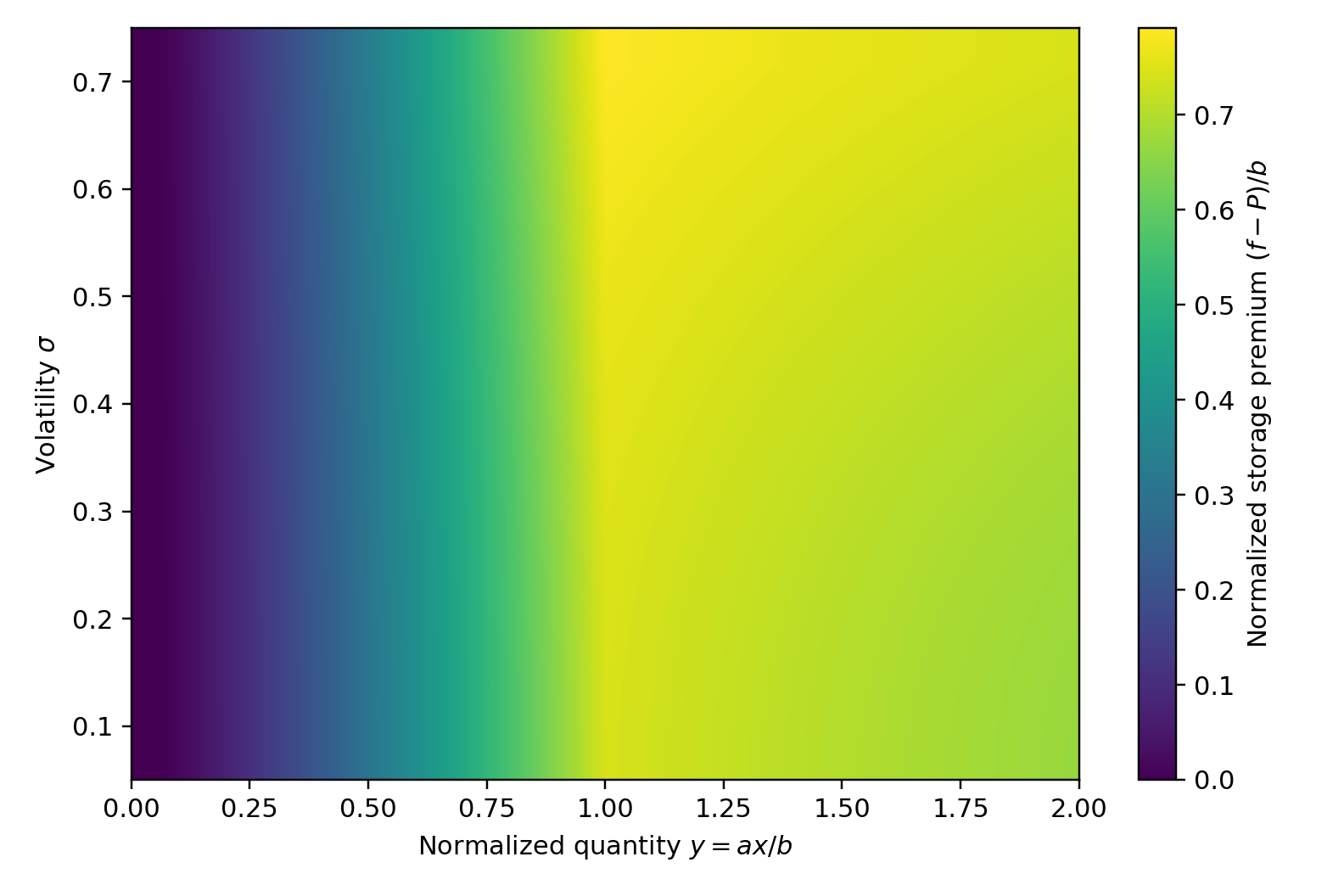}
\caption{Normalized storage premium $(f-P)/b$ over quantity and volatility.}\label{fig:premiumheat}
\end{figure}

Table~\ref{tab:sens} summarizes several one-at-a-time sensitivity experiments. Discounting has the clearest influence on the selling threshold. Increasing $r$ from $0.02$ to $0.10$ moves the normalized boundary from about $0.013$ to $0.087$ and reduces the premium at $y=1$. A high discount rate penalizes waiting and therefore makes immediate sale optimal over a larger range of states. Higher depreciation has a different effect in this model. It accelerates the reduction of the quantity state and makes a future high-price state easier to reach; the storage premium therefore increases modestly with $\delta$. A larger $m$ works in the opposite direction by supporting the quantity state and reducing the value of waiting.

\begin{table}[H]
\centering
\caption{One-at-a-time comparative statics. Premium is reported at $y=1$ and normalized by $b$.}\label{tab:sens}
\begin{tabular}{llrrr}
\toprule
Parameter & Value & $y^*$ & $p^*/b$ & $(f-P)/b$ at $y=1$\\
\midrule
$r$ & 0.02 & 0.013 & 0.987 & 0.848\\
$r$ & 0.05 & 0.047 & 0.953 & 0.750\\
$r$ & 0.10 & 0.087 & 0.913 & 0.651\\
\addlinespace
$\delta$ & 0.08 & 0.047 & 0.953 & 0.739\\
$\delta$ & 0.15 & 0.047 & 0.953 & 0.750\\
$\delta$ & 0.25 & 0.047 & 0.953 & 0.763\\
\addlinespace
$m$ & $-0.05$ & 0.040 & 0.960 & 0.762\\
$m$ & 0.00 & 0.047 & 0.953 & 0.750\\
$m$ & 0.05 & 0.047 & 0.953 & 0.737\\
\addlinespace
$\sigma$ & 0.10 & 0.047 & 0.953 & 0.745\\
$\sigma$ & 0.25 & 0.047 & 0.953 & 0.750\\
$\sigma$ & 0.50 & 0.040 & 0.960 & 0.767\\
\bottomrule
\end{tabular}
\end{table}

Figure~\ref{fig:rdelta} combines $r$ and $\delta$. In the baseline range, the normalized threshold is dominated by discounting and is comparatively insensitive to depreciation. This does not mean that $\delta$ is economically irrelevant: Table~\ref{tab:sens} shows that it affects the continuation value even when the threshold changes little.

\begin{figure}[H]
\centering
\includegraphics[width=.80\textwidth]{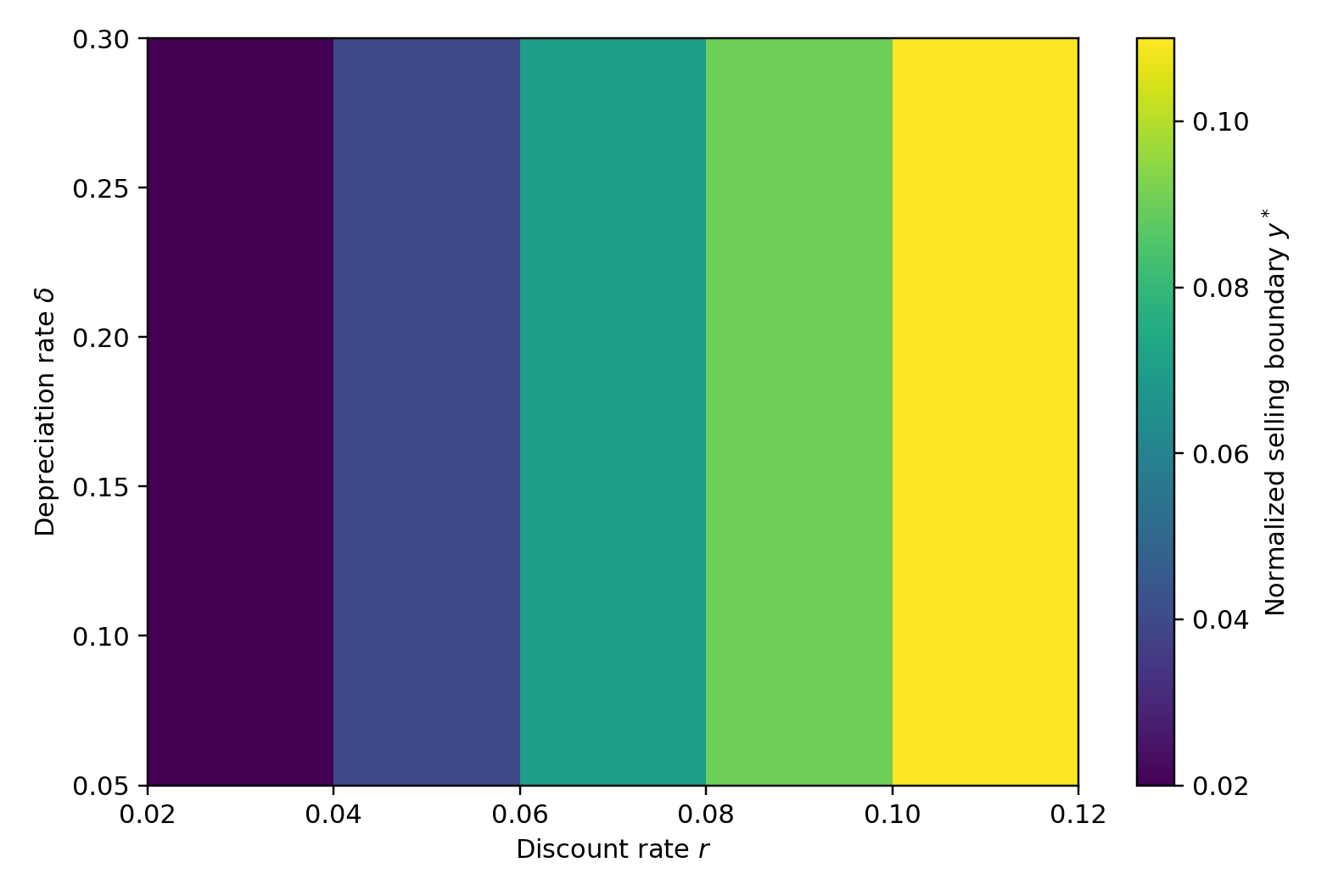}
\caption{Normalized selling boundary over discount and depreciation rates.}\label{fig:rdelta}
\end{figure}

\subsection{Demand scaling and elasticity}

The old numerical study suggested that changing $a$ and $b$ mainly rescales the equilibrium. Figure~\ref{fig:scaling} makes that observation precise numerically. Four very different pairs $(a,b)$ collapse onto the same curve after the normalization $y=ax/b$ and $f/b$. The corresponding normalized threshold is approximately $0.045$ in all four cases and $p^*/b\approx0.955$.

\begin{figure}[H]
\centering
\includegraphics[width=.84\textwidth]{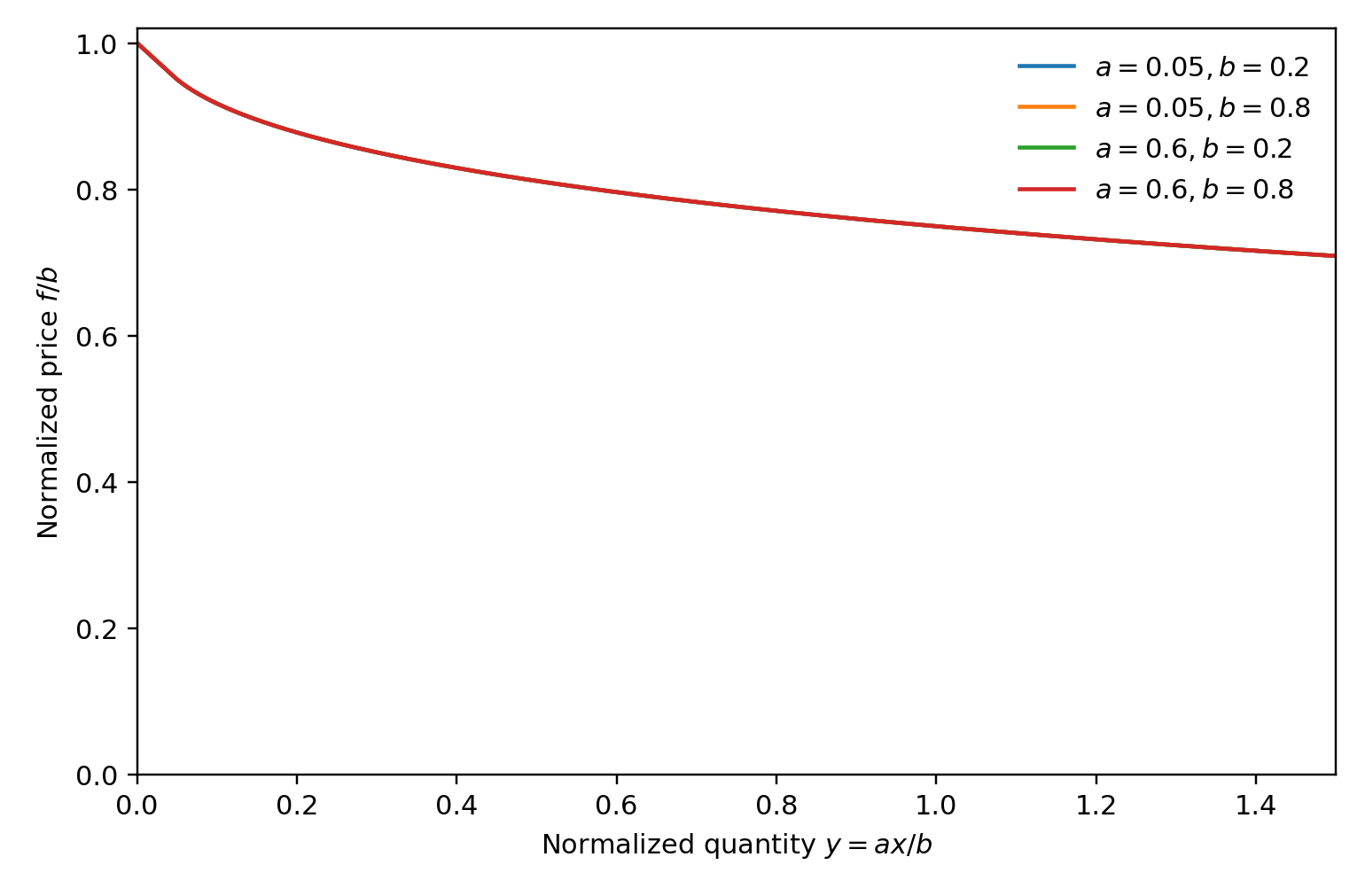}
\caption{Scaling collapse of the equilibrium for four pairs $(a,b)$.}\label{fig:scaling}
\end{figure}

The scaling can also be seen analytically. Let $x=(b/a)y$ and $f(x)=bF(y)$. Then $Q(f(x))=(b/a)(1-F(y))$ whenever $0\le F\le1$, and the state dynamics in normalized units depend on $r,\delta,m,\sigma$ but not separately on $a$ and $b$. The obstacle becomes $(1-y)^+$. Thus the linear-demand parameters determine the units of quantity and price, while the normalized equilibrium shape is driven primarily by the dynamic parameters. This property is useful when comparing commodities whose demand curves differ mostly by scale.

\subsection{Monte Carlo validation}

The free-boundary representation gives an independent way to check the numerical equilibrium. For $x>x^*$ we simulate the equilibrium diffusion and stop at the first passage time \eqref{eq:taustar}. The Monte Carlo estimator is
\[
\widehat f_{MC}(x)=\frac1M\sum_{j=1}^{M}e^{-r\tau_j^*}p^*.
\]
The simulation uses a positivity-preserving log-Euler step with the equilibrium drift frozen over each time interval. Table~\ref{tab:mc} compares this estimator with the finite-difference equilibrium. The largest absolute difference among the four reported states is below $4.3\times10^{-4}$. The differences are small relative to the value level and consistent with the combined Monte Carlo and time-discretization error.

\begin{table}[H]
\centering
\caption{Monte Carlo validation of the equilibrium value. Standard errors refer to the Monte Carlo estimator.}\label{tab:mc}
\begin{tabular}{rrrrrr}
\toprule
$x_0$ & $y_0$ & Grid value & MC value & MC s.e. & Abs. difference\\
\midrule
0.5 & 0.125 & 0.181052 & 0.180739 & $3.32\times10^{-5}$ & $3.13\times10^{-4}$\\
1.0 & 0.250 & 0.172611 & 0.172235 & $4.99\times10^{-5}$ & $3.75\times10^{-4}$\\
4.0 & 1.000 & 0.149885 & 0.149459 & $9.71\times10^{-5}$ & $4.26\times10^{-4}$\\
8.0 & 2.000 & 0.135668 & 0.135387 & $1.24\times10^{-4}$ & $2.81\times10^{-4}$\\
\bottomrule
\end{tabular}
\end{table}

\begin{figure}[H]
\centering
\includegraphics[width=.74\textwidth]{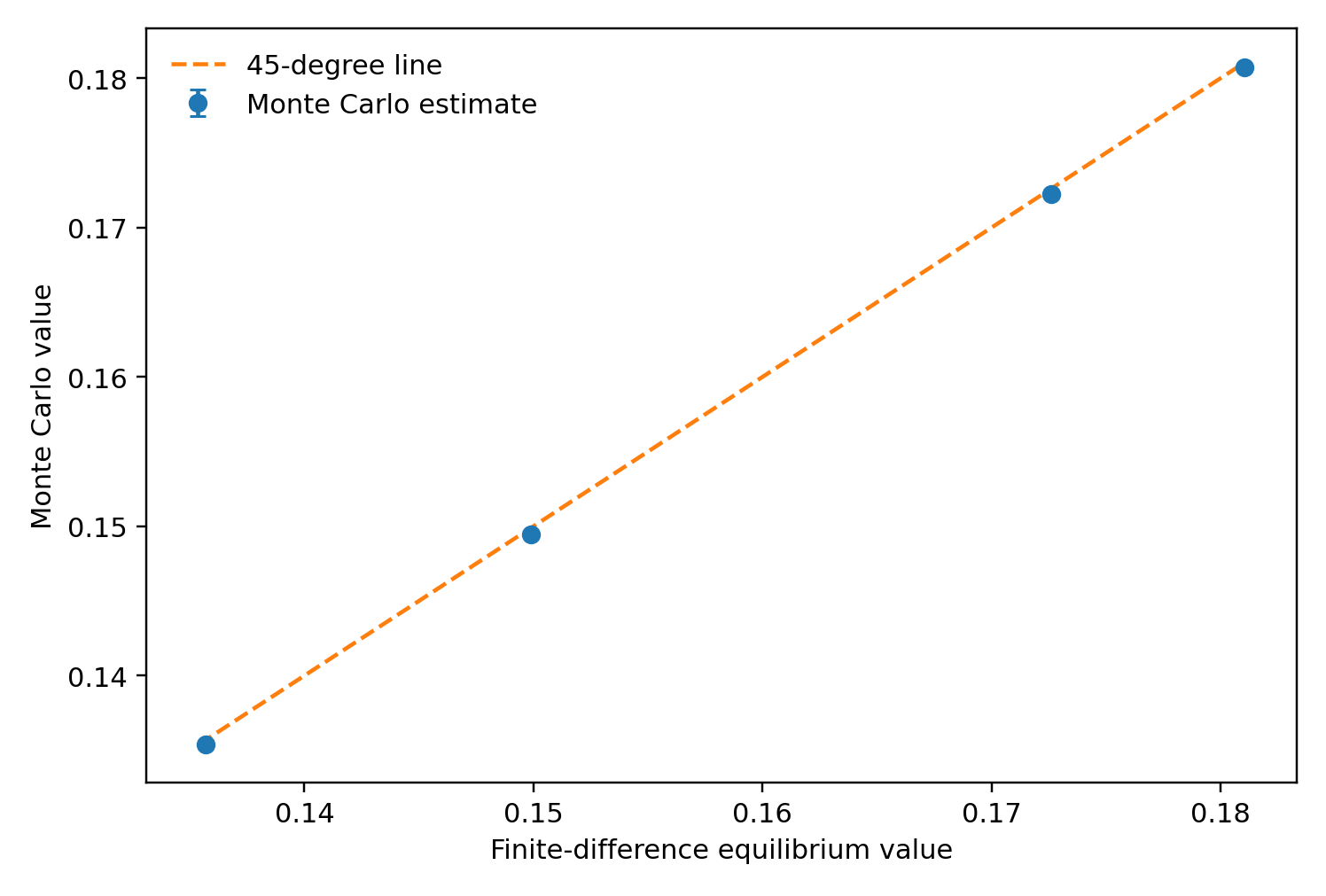}
\caption{Finite-difference values versus Monte Carlo first-passage estimates. Error bars are approximate 95\% Monte Carlo intervals.}\label{fig:mc}
\end{figure}

This validation is especially useful in the right tail. A zero boundary at $X_{\max}$ mechanically prevents the numerical solution from representing paths that start near the truncation point and later return to the selling region. The first-passage estimator has no such restriction when the simulated state is allowed to move beyond the finite-difference grid and the equilibrium function is continued with the power-law asymptotic form. The agreement in Table~\ref{tab:mc} therefore provides direct support for the far-field treatment.

\subsection{Distribution of the optimal selling time}

The equilibrium provides more than a price curve. Once $x^*$ is known, it gives a distribution for the time at which the commodity is optimally sold. Figure~\ref{fig:survival} shows empirical survival functions of $\tau^*$ for a common normalized initial state $y_0=0.5$ and three volatility levels. Table~\ref{tab:hittime} reports selected quantiles.

\begin{figure}[H]
\centering
\includegraphics[width=.80\textwidth]{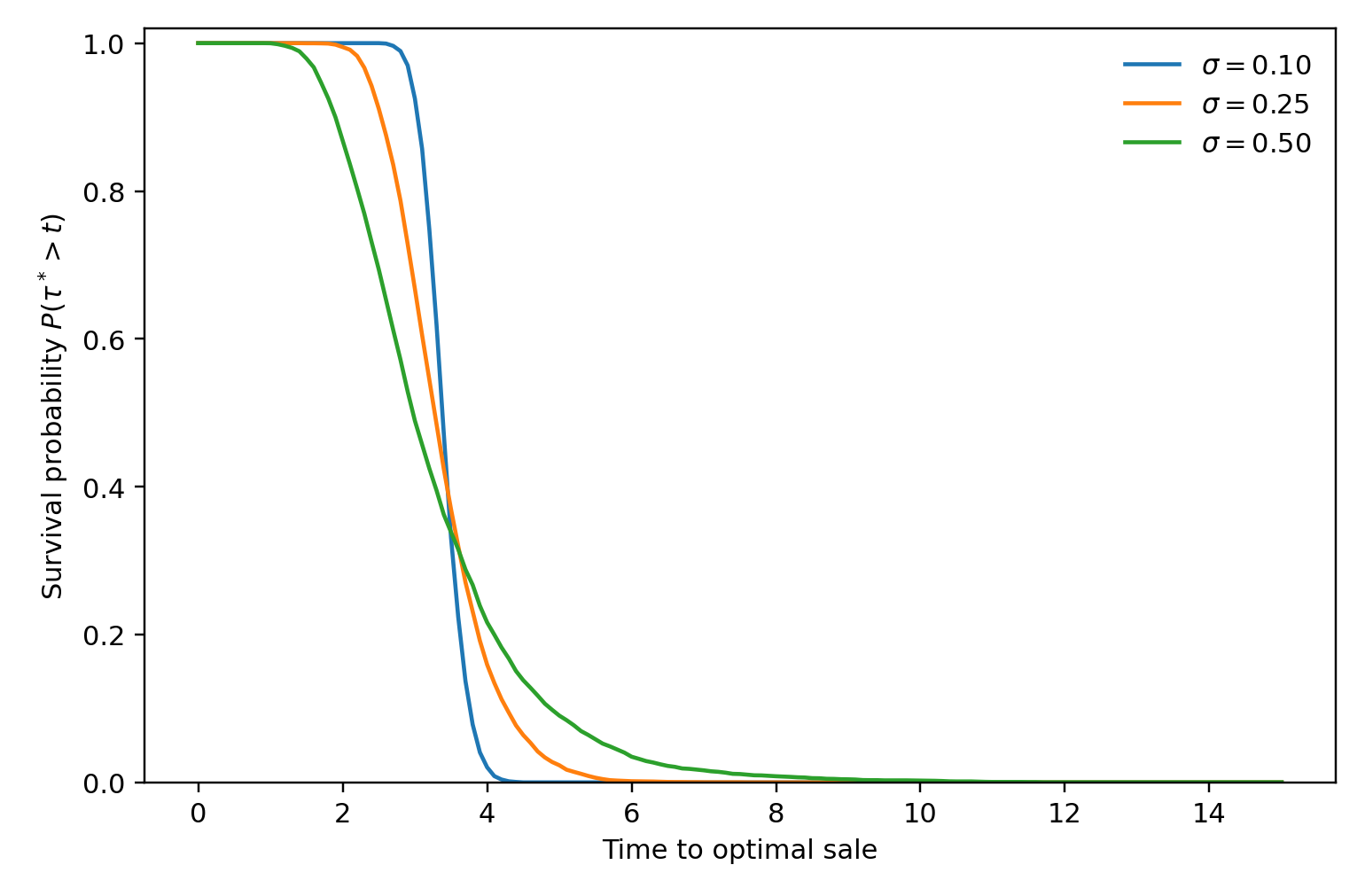}
\caption{Empirical survival probability of the optimal selling time from $y_0=0.5$.}\label{fig:survival}
\end{figure}

\begin{table}[H]
\centering
\caption{Monte Carlo summary of the optimal selling time from $y_0=0.5$.}\label{tab:hittime}
\begin{tabular}{rrrrrr}
\toprule
$\sigma$ & Mean & 25\% & Median & 75\% & 90\%\\
\midrule
0.10 & 3.393 & 3.198 & 3.381 & 3.571 & 3.762\\
0.25 & 3.353 & 2.865 & 3.270 & 3.754 & 4.262\\
0.50 & 3.246 & 2.349 & 2.968 & 3.865 & 4.873\\
\bottomrule
\end{tabular}
\end{table}

The mean selling time declines moderately with volatility, but the more visible effect is dispersion. At $\sigma=0.10$ the central part of the distribution is narrow. At $\sigma=0.50$ the lower quartile occurs much earlier, while the upper quantiles extend considerably farther into the future. This is useful economically: volatility increases the option value not because every path waits longer, but because it creates a wider set of possible timing outcomes and more upside from favorable state movements.

\section{Discussion}

The revised formulation changes several interpretations of the old numerical manuscript. The first concerns existence. A fixed point on a truncated interval is not automatically the restriction of a fixed point computed on a larger interval, because the value inside the smaller interval may depend on paths that leave it. The compact-open proof of Theorem~\ref{thm:existence} removes that issue. It also makes clear why a common Lipschitz bound is useful: it provides both SDE regularity and the compactness needed by the global fixed-point theorem.

The second concerns uniqueness. Existence and uniqueness should not be combined in a theorem unless both are actually proved under the stated assumptions. Theorem~\ref{thm:existence} therefore makes no uniqueness statement. If a parameter set lies in a contraction regime, uniqueness follows from the contraction result, and the outer iteration is justified by Banach's theorem. If the contraction condition is not available, the existence result remains valid under $m\le r+\delta$, but the possibility of multiple fixed points should not be ruled out without further analysis.

The third concerns numerical convergence. The number of outer iterations is not a numerical order of accuracy. The spatial convergence experiment in Table~\ref{tab:grid} measures the approximation of the obstacle operator and is approximately first order. The residual plot in Figure~\ref{fig:residual} measures the convergence of the fixed-point iteration. Proposition~\ref{prop:error} explains how the two errors combine. This separation is important when numerical results are used as evidence for a theorem or when a tolerance is selected for practical computation.

The fourth concerns state truncation. The storage option is perpetual and the demand diffusion can return from a large state to the selling region. A zero upper boundary therefore has an economic meaning that is not part of the original model: it says that the storage option becomes worthless exactly at the chosen truncation point. The asymptotic Robin condition avoids imposing that extra assumption. The large difference in Table~\ref{tab:boundary} shows that this is not only a technical refinement.

Finally, the expanded numerical analysis clarifies which economic parameters drive which features of the equilibrium. The linear demand parameters $a$ and $b$ mainly set the price and quantity scales. Discounting determines the willingness to wait and therefore has a strong effect on the selling boundary. Depreciation and drift change the speed at which the state moves toward or away from the high-price region. Volatility raises the storage premium and broadens the selling-time distribution. These distinctions are more informative than comparing a few equilibrium curves without normalizing the state or reporting the threshold.

\section{Conclusion}\label{sec:conclusion}

We have developed a continuous-time model of storable commodity prices in which the price is generated by the storage decision rather than imposed as an exogenous stochastic process. A candidate speculative price determines the demand diffusion, and the perpetual optimal-stopping value of that diffusion determines a new price function. The fixed point of this feedback is called an endogenous storage-option equilibrium (ESOE) in this paper.

The main theoretical result is a global existence theorem on a compact admissible class. The proof uses the topology of locally uniform convergence and avoids an interval-by-interval construction of local fixed points. The storage operator preserves boundedness, monotonicity, and the natural Lipschitz constant of the inverse-demand curve when $m\le r+\delta$. Continuity follows from SDE stability and an exponentially small infinite-horizon tail. Schauder--Tychonoff then gives existence. Uniqueness is deliberately kept separate and follows only when an additional contraction condition is available.

The equilibrium is a nonlinear free-boundary problem. We approximate it by solving a stationary linear obstacle problem for each current outer iterate, using a monotone finite-difference operator and policy iteration. An asymptotic Robin condition replaces the artificial zero upper boundary and preserves the slow decay of the perpetual waiting value. The numerical experiments distinguish spatial convergence from fixed-point convergence and show stable first-order behavior. Monte Carlo first-passage values provide an independent validation.

The comparative statics show that the storage premium is strongly affected by discounting and volatility, while the linear inverse-demand parameters largely act through simple scale changes. The selling threshold gives an immediate economic interpretation of the equilibrium, and the first-passage distribution adds information about the timing risk of optimal sales. These features make the model useful not only as a theoretical continuous-time extension of competitive storage, but also as a practical framework for studying endogenous commodity prices and storage decisions.

\appendix

\section{Details of the continuity argument}\label{app:continuity}

This appendix gives additional details for Lemma~\ref{lem:Tcont}. Let $K,T>0$ and let $h_n\to h$ locally uniformly in $\mathcal{H}$. Because all functions in $\mathcal{H}$ are $a$-Lipschitz and bounded by $b$, the drifts satisfy a uniform global Lipschitz estimate
\[
|\mu_{h_n}(x)-\mu_{h_n}(y)|\le L_\mu|x-y|,
\]
with $L_\mu$ independent of $n$. They also satisfy a common linear-growth bound. Moreover, $\mu_{h_n}\to\mu_h$ uniformly on every compact interval.

Let $X^{n,x}$ and $X^x$ denote the solutions driven by the same Brownian motion. For $R>K$ define the stopping time
\[
\eta_R=\inf\{t\ge0:X_t^{n,x}\vee X_t^x\ge R\}.
\]
On $[0,T\wedge\eta_R]$, standard SDE estimates and Gronwall's inequality give
\[
\sup_{0\le x\le K}\mathbb{E}\left[
\sup_{0\le t\le T\wedge\eta_R}|X_t^{n,x}-X_t^x|^2\right]
\longrightarrow0.
\]
Uniform moment bounds and Markov's inequality show that the probability of $\eta_R\le T$ can be made uniformly small by first taking $R$ large. Hence the same convergence holds in probability on $[0,T]$, uniformly for $x\in[0,K]$.

For the finite-horizon stopping problems, one may use the Snell-envelope stability theorem for bounded continuous rewards. A direct discrete-time approximation gives the same conclusion: approximate the set of stopping times by a fine deterministic time grid, use convergence of the finite-dimensional distributions and boundedness of $P$, and then let the time mesh go to zero. The result is
\[
\sup_{0\le x\le K}|(\mathcal{T}_T h_n)(x)-(\mathcal{T}_T h)(x)|\to0.
\]

It remains to pass from finite to infinite horizon. For any stopping time $\tau$, split the payoff according to $\{\tau\le T\}$ and $\{\tau>T\}$. The first part is bounded by the finite-horizon value, while the second is at most $be^{-rT}$. Thus
\[
\mathcal{T}_T h\le\mathcal{T} h\le\mathcal{T}_T h+be^{-rT},
\]
uniformly in $h\in\mathcal{H}$ and $x\ge0$. The same estimate holds for $h_n$. Therefore
\begin{align*}
\sup_{0\le x\le K}|\mathcal{T} h_n-\mathcal{T} h|
&\le 2be^{-rT}
+\sup_{0\le x\le K}|\mathcal{T}_T h_n-\mathcal{T}_T h|.
\end{align*}
First choose $T$ large and then $n$ large.

\section{Discrete obstacle operator}\label{app:fd}

For completeness, the coefficients of the tridiagonal continuation operator are written explicitly. Let $M_h=rI-\mathcal A_h^{\Delta x}$. For an interior node $i$, define $d_i=\frac12\sigma^2x_i^2/(\Delta x)^2$.

If $\mu_i\ge0$, the backward Kolmogorov generator is discretized with a forward difference for the first derivative:
\[
(\mathcal A_h^{\Delta x}V)_i
=d_iV_{i-1}+\left(-2d_i-\frac{\mu_i}{\Delta x}\right)V_i
+\left(d_i+\frac{\mu_i}{\Delta x}\right)V_{i+1}.
\]
Hence
\[
(M_hV)_i=-d_iV_{i-1}
+\left(r+2d_i+\frac{\mu_i}{\Delta x}\right)V_i
-\left(d_i+\frac{\mu_i}{\Delta x}\right)V_{i+1}.
\]

If $\mu_i<0$, a backward difference gives
\[
(\mathcal A_h^{\Delta x}V)_i
=\left(d_i-\frac{\mu_i}{\Delta x}\right)V_{i-1}
+\left(-2d_i+\frac{\mu_i}{\Delta x}\right)V_i+d_iV_{i+1},
\]
and therefore
\[
(M_hV)_i=-\left(d_i-\frac{\mu_i}{\Delta x}\right)V_{i-1}
+\left(r+2d_i-\frac{\mu_i}{\Delta x}\right)V_i-d_iV_{i+1}.
\]
In both cases the off-diagonal entries of $M_h$ are non-positive and the diagonal is positive. The row is strictly diagonally dominant by $r>0$. This sign structure is the reason for the apparently reversed upwind choice relative to a forward transport equation: the PDE here is a backward generator equation.

The Robin condition \eqref{eq:robin} is discretized as
\[
X_{\max}\frac{V_N-V_{N-1}}{\Delta x}=\xi_-V_N,
\]
or
\[
-V_{N-1}+\left(1-\xi_-\frac{\Delta x}{X_{\max}}\right)V_N=0.
\]
Since $\xi_-<0$, the coefficient of $V_N$ is larger than one.

Policy iteration selects one of two equations at every interior node:
\[
V_i=P_i\quad\text{or}\quad (M_hV)_i=0.
\]
The selected linear system is solved, the complementarity residual is recomputed, and the policy is updated. Once the active set no longer changes, the discrete obstacle problem has been solved to the linear-system tolerance.

\section{Additional numerical tables}\label{app:tables}

Table~\ref{tab:sigextra} reports the wider volatility experiment used in Figures~\ref{fig:sigmathreshold} and \ref{fig:premiumheat}. The boundary is grid-valued, so small movements should be interpreted together with the smoother premium column.

\begin{table}[H]
\centering
\caption{Extended volatility sensitivity.}\label{tab:sigextra}
\begin{tabular}{rrrrr}
\toprule
$\sigma$ & $y^*$ & $p^*/b$ & Premium at $y=1$ & Outer iterations\\
\midrule
0.05 & 0.045 & 0.955 & 0.7443 & 11\\
0.10 & 0.045 & 0.955 & 0.7450 & 11\\
0.20 & 0.045 & 0.955 & 0.7477 & 11\\
0.30 & 0.045 & 0.955 & 0.7522 & 11\\
0.40 & 0.040 & 0.960 & 0.7585 & 12\\
0.50 & 0.040 & 0.960 & 0.7663 & 12\\
0.60 & 0.040 & 0.960 & 0.7756 & 12\\
0.70 & 0.035 & 0.965 & 0.7859 & 12\\
0.75 & 0.035 & 0.965 & 0.7913 & 12\\
\bottomrule
\end{tabular}
\end{table}

\end{document}